\documentclass[sigconf,nonacm]{acmart}

\usepackage[ruled,vlined,linesnumbered]{algorithm2e} 
\usepackage{multirow,graphicx,float,subfig,caption,fancyhdr,xcolor,enumitem,colortbl}
\usepackage{amsfonts}
\usepackage{epstopdf}
\usepackage{bbding}
\usepackage{array}
\usepackage{hyperref}
\usepackage{pifont}
\usepackage[most]{tcolorbox}
\usepackage[utf8]{inputenc}
\usepackage{lipsum} 
\setcopyright{none} 
\SetKwProg{Fn}{Function}{:}{}
\SetKwProg{Proc}{Procedure}{:}{}
\SetKwFunction{brb}{BranchB}
\SetKwFunction{brp}{BranchP}
\SetKwFunction{brpplus}{BranchP$^*$}
\SetKwFunction{upd}{Update}
\SetKwFunction{ub}{UpperBound}
\SetKwFunction{heu}{Heuristic}
\SetKwFunction{exp}{Expand}
\SetKwFunction{mdbb}{MDBB}
\SetKwFunction{mdbp}{MDBP}
\SetKwFunction{mdbase}{MDBase}
\SetKwFunction{mdcadp}{MDCadp}
\SetKwFunction{pivotplus}{Pivot+}
\SetKwFunction{mdbbnh}{MDBB-H}
\SetKwFunction{mdbpnh}{MDBP-H}
\SetKwFunction{mdbbnu}{MDBB-U}
\SetKwFunction{mdbpnu}{MDBP-U}
\SetKwFunction{mbc}{MBC}
\SetKwFunction{fraudar}{FRAUDAR}
\SetKwFunction{caregnn}{CARE-GNN}
\SetKwFunction{redg}{ReduceG}
\SetKwFunction{redi}{Reduce}
\SetKwFunction{cnred}{CnReduction}
\SetKw{Break}{break}
\SetKw{Continue}{continue}
\SetKw{And}{and}
\SetKw{Or}{or}
\SetKw{In}{in}

\setcopyright{acmlicensed}
\copyrightyear{2026}
\acmYear{2026}
\acmDOI{XXXXXXX.XXXXXXX}
\acmConference[SIGMOD '26]{-}{May 31--June 05,
	2026}{Bengaluru, India}

\newtheoremstyle{proofs}{}{}{}{}{\scshape}{.}{ }{}

\theoremstyle{proofs}

\newtheorem{theorem}{\hspace{1em}Theorem}
\newtheorem{lemma}{\hspace{1em}Lemma}

\newenvironment{proofsketch}{\begin{proof}}{\end{proof}}

\theoremstyle{definition}
\newtheorem*{problem}{Problem statement}
\newtheorem*{detail}{Implementation details}

\theoremstyle{remark}

\let\savedbaselinestretch\baselinestretch

\setlist[itemize]{leftmargin=2em}
\setlist[enumerate]{leftmargin=2em}

\begin{document}
	
	\author{Donghang Cui}
	\affiliation{%
		\institution{Beijing Institute of Technology}
		\country{}
	}
	\email{cuidonghang@bit.edu.cn}
	
	\author{Ronghua Li}
	\affiliation{%
		\institution{Beijing Institute of Technology}
		\country{}
	}
	\email{lironghuabit@126.com}
	
	\author{Qiangqiang Dai}
	\affiliation{%
		\institution{Beijing Institute of Technology}
		\country{}
	}
	\email{qiangd66@gmail.com}
	
	\author{Hongchao Qin}
	\affiliation{%
		\institution{Beijing Institute of Technology}
		\country{}
	}
	\email{qhc.neu@gmail.com}
	
	\author{Guoren Wang}
	\affiliation{%
		\institution{Beijing Institute of Technology}
		\country{}
	}
	\email{wanggrbit@gmail.com}
	
	\title{On the Efficient Discovery of Maximum $k$-Defective Biclique}
	
	\begin{abstract}
		
		The problem of identifying the maximum edge biclique in bipartite graphs has attracted considerable attention in bipartite graph analysis, with numerous real-world applications such as fraud detection, community detection, and online recommendation systems. However, real-world graphs may contain noise or incomplete information, leading to overly restrictive conditions when employing the biclique model. To mitigate this, we focus on a new relaxed subgraph model, called the $k$-defective biclique, which allows for up to $k$ missing edges compared to the biclique model. We investigate the problem of finding the maximum edge $k$-defective biclique in a bipartite graph, and prove that the problem is NP-hard. To tackle this computation challenge, we propose a novel algorithm based on a new branch-and-bound framework, which achieves a worst-case time complexity of $O(m\alpha_k^n)$, where $\alpha_k < 2$. We further enhance this framework by incorporating a novel pivoting technique, reducing the worst-case time complexity to $O(m\beta_k^n)$, where $\beta_k < \alpha_k$. To improve the efficiency, we develop a series of optimization techniques, including graph reduction methods, novel upper bounds, and a heuristic approach. Extensive experiments on 10 large real-world datasets validate the efficiency and effectiveness of the proposed approaches. The results indicate that our algorithms consistently outperform state-of-the-art algorithms, offering up to $1000\times$ speedups across various parameter settings.

	\end{abstract}

	\maketitle

	\section{Introduction}
	
	Bipartite graph is a commonly encountered graph structure consisting of two disjoint vertex sets and an edge set, where each edge connects vertices from the two different sets. In recent decades, bipartite graphs have been utilized as a fundamental tool for modeling various binary relationships, such as user-product interactions in e-commerce \cite{RN50}, authorship connections between authors and publications \cite{RN34}, and associations between genes and phenotypes in bioinformatics \cite{RN39,RN38}. Mining cohesive subgraphs from bipartite graphs has proven to be critically important across diverse domains \cite{RN66,RN67}. For instance, identifying cohesive subgraphs in reviewer-product graphs facilitates the identification of fraudulent reviews or camouflage attacks \cite{RN65}.
	
	
	The biclique, which connects all vertices on both sides, serves as a fundamental cohesive subgraph structure \cite{RN80,RN81,RN15} due to its wide applications \cite{RN68,RN58,RN57,RN69}. 
	Extensive methods for maximum biclique search have been developed in recent years \cite{RN47,RN58,RN66,RN14,RN22,RN13}. 
	However, real-world graphs often contain noisy or incomplete information \cite{RN70}, making the biclique model overly restrictive for capturing cohesive subgraphs on such a kind of bipartite graphs. 	
	More relaxed biclique models offer a promising alternative to address this challenge. The $k$-biplex model, which allows each vertex to have at most $k$ non-neighbors, has recently attracted significant attention \cite{RN66,RN71,RN72}. However, due to its excessive structural relaxation, a slightly increase in $k$ can lead to: (1) a sharp rise in quantity, potentially reducing the mining efficiency; and (2) a significant decrease in density, undermining the capability to represent cohesive communities.
	
	
	\begin{figure*}[!t]

		\setlength{\belowcaptionskip}{-0.3cm} 
		\subfloat[\centering{Example graph showcasing a biclique and two $k$-defective cliques}]{
			\includegraphics[width=.3\linewidth]{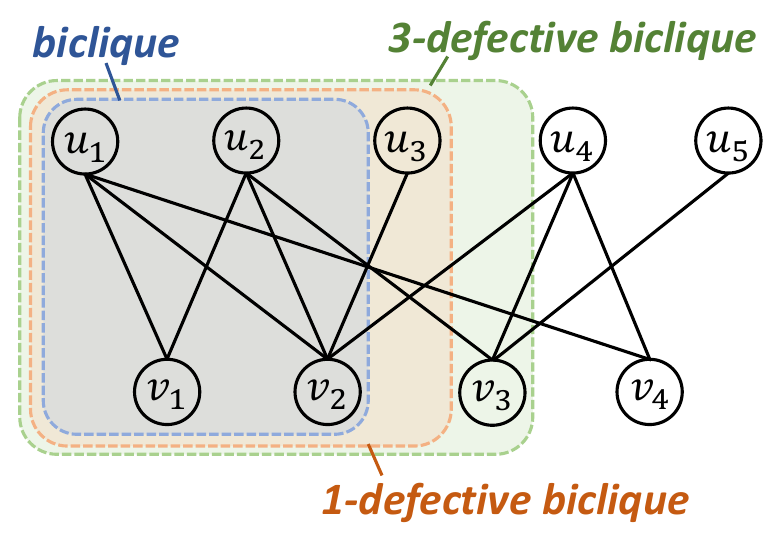}
			\label{fig:eg1a}
		}\hspace{5mm}
		\subfloat[\centering{Example of incorrect output from $k$-defective clique algorithms}]{
			\includegraphics[width=.28\linewidth]{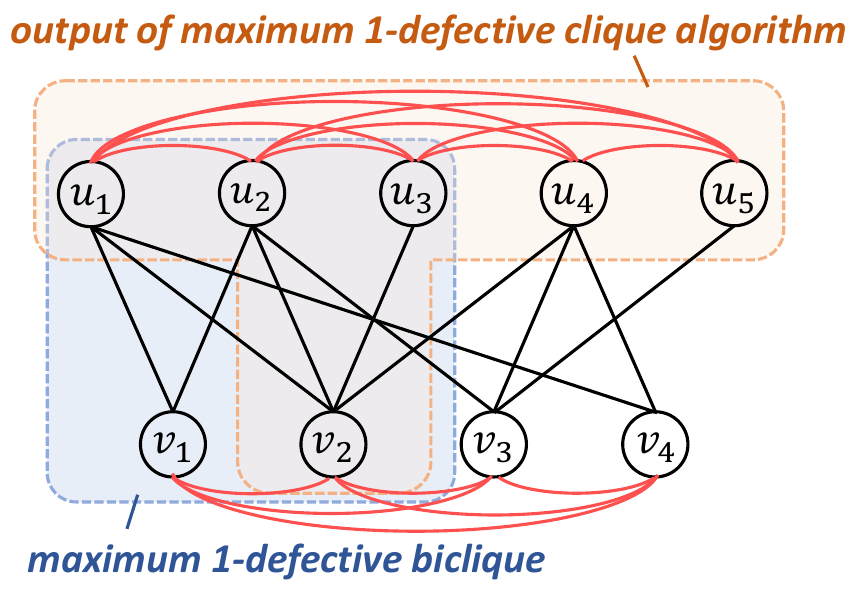}
			\label{fig:eg1b}
		}\hspace{5mm}
		\subfloat[\centering{Comparison of quantity and density between $k$-biplex and $k$-defective biclique on \textit{Amazon} dataset with side size $\ge14$}]{
			\includegraphics[width=.32\linewidth, height=3.5cm]{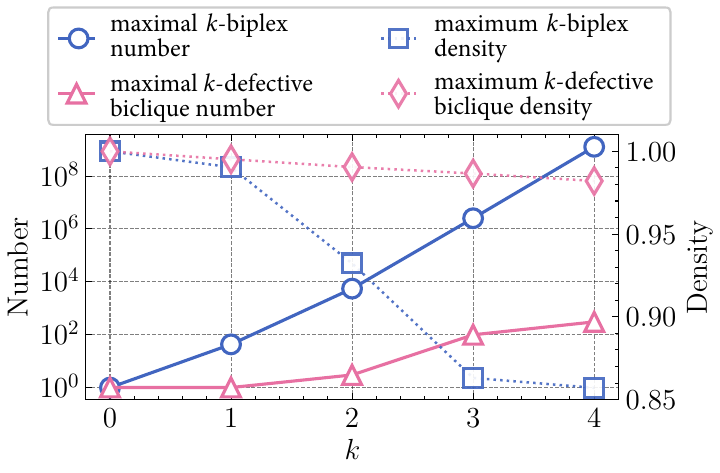}
			\label{fig:eg1c}
		}
		\caption{\centering {Illustrations of $k$-defective biclique: examples and comparisons}}
		\label{fig:eg1}
	\end{figure*}

	Inspired by the widely studied $k$-defective clique in traditional graphs \cite{RN8,RN6,RN2}, we introduce a novel cohesive subgraph model for bipartite graphs, named $k$-defective biclique. Specifically, a $k$-defective biclique is a subgraph of a bipartite graph containing at most $k$ missing edges. Fig. \ref{fig:eg1a} illustrates an example graph showcasing two $k$-defective bicliques, where the $1$-defective biclique has one missing edge $(u_3,v_1)$, and the $3$-defective biclique has three missing edges $(u_1,v_3)$, $(u_3,v_1)$ and $(u_3,v_3)$.
	
	The $k$-defective biclique offers greater flexibility over biclique, as its density can be controlled by adjusting $k$ to suit specific requirements. When $k=0$, the $k$-defective biclique degenerates to the biclique, thereby illustrating its generalization capability. 
	In comparison to the $k$-biplex, $k$-defective biclique offers a more fine-grained relaxation of the biclique due to its stricter retention of edges. As illustrated in Fig. \ref{fig:eg1c}, the $k$-defective biclique exhibits much closer quantity and density to the biclique compared to $k$-biplex, especially for larger $k$. Here, we refer to a $k$-defective biclique as "maximal" if it cannot be contained within any other $k$-defective biclique (or $k$-biplex), and as "maximum" when it is the maximal $k$-defective biclique (or $k$-biplex) with the largest number of edges. 
	In this paper, we focus mainly on the problem of identifying a maximum $k$-defective biclique. We prove that the maximum $k$-defective biclique search problem is NP-hard, indicating that no polynomial-time algorithms exist for solving this problem unless NP$=$P. 

	Notably, adapting existing maximum/maximal $k$-defective clique algorithms \cite{RN4,RN2} and maximum biclique algorithms \cite{RN13,RN14} to our problem poses significant challenges for the following reasons: (1) Adapting the maximum $k$-defective clique algorithm involves fully connecting vertices on each side of the bipartite graph (as shown in Fig. \ref{fig:eg1b}). However, applying such algorithms to these graphs may yield incorrect results. For instance, in Fig. \ref{fig:eg1b}, the maximum $1$-defective biclique has 5 edges whereas the output of maximum $1$-defective clique algorithms yields only 4 edges; (2) Applying maximal $k$-defective clique algorithms on such each-side-fully-connected graphs will generate a large number of redundant branches and is computationally expensive (as evidenced in Table \ref{tab:rt}); and (3) Adapting maximum biclique algorithms to our problem involves a preprocess of adding $k$ edges on the original graph, which generates in $\binom{\overline m}{k}$ possible input graphs where $\overline m$ denotes the number of missing edges. Searching for the maximum biclique across all such input graphs remains highly inefficient, making this approach impractical for large-scale graphs.

	\noindent\textbf{Contributions.} To efficiently address the maximum $k$-defective biclique search problem, we conduct an in-depth study of its properties and develop two novel algorithms with the worst-case time complexity less than $O^*(2^n)$. We also present a series of optimization techniques to further improve the efficiency of our algorithms. Finally, we evaluate the effectiveness and efficiency of our methods through comprehensive experiments. The main contributions of this paper are summarized as follows:

	\noindent\underline{\emph{A new cohesive subgraph model.}} We propose a new cohesive subgraph model for bipartite graphs named the $k$-defective biclique, which is defined as a subgraph with at most $k$ missing edges. We then demonstrate that the problem of finding the maximum $k$-defective biclique in bipartite graphs is NP-hard.
	
	
	\noindent\underline{\emph{Two Novel Search Algorithms.}} We propose two novel maximum $k$-defective biclique search algorithms with nontrivial worst-case time complexity guarantees. The first algorithm utilizes a branch-and-bound technique combined with a newly-developed binary branching strategy, which has a time complexity of $O(m\alpha_k^n)$, where $n$ and $m$ represent the number of vertices and edges of the bipartite graph respectively, and $\alpha_k$ is a positive real number strictly less than 2 (e.g., for $k=1, 2, 3$, $\alpha_k = 1.911, 1.979,$ and $1.995$, respectively). The second algorithm employs an innovative pivoting-based branching strategy that effectively integrates with our binary branching strategy, which achieves a lower time complexity of $O(m\beta_k^n)$, where $\beta_k$ is a positive real number strictly less than $\alpha_k$. (e.g., when $k=1, 2, 3$, $\beta_k = 1.717, 1.856,$ and $1.931$, respectively).

	\noindent\underline{\emph{Nontrivial optimization techniques.}} We also develop a series of nontrivial optimization techniques to improve the efficiency of our proposed algorithms, including graph reduction techniques, novel upper-bound techniques, and a heuristic approach. Specifically, the graph reduction techniques focus on reducing the size of the input bipartite graph based on constraints, such as the core and the number of common neighbors. We devise tight vertex and edge upper bounds to minimize unnecessary branching, and then design an efficient algorithm to compute the proposed upper bounds with linear time complexity. Additionally, we propose a heuristic approach by incrementally expanding subsets in a specific order to obtain an initial large $k$-defective biclique, which helps filter out smaller $k$-defective bicliques from further computation.
	
	
	\noindent\underline{\emph{Comprehensive experimental studies.}} We conduct extensive experiments to evaluate the efficiency, effectiveness, and scalability of our algorithms on 11 large real-world graphs. The experimental results show that our algorithms consistently outperform the state-of-the-art baseline method \cite{RN64} across all parameter settings. In particular, our algorithm achieves improvements of up to three orders of magnitude on most datasets. For instance, on the \textit{Twitter} (9M vertices, 10M edges) dataset with $k=2$, our algorithm takes 10 seconds while the baseline method cannot finish within 3 hours. Additionally, we perform a case study on a fraud detection task to show the high effectiveness of our solutions in real-world applications.
	

	\section{Preliminaries} \label{sec:pre}
	Let $G=(U,V,E)$ be an undirected bipartite graph, where $U$ and $V$ are two disjoint sets of vertices, and $E\subseteq U\times V$ is the set of edges. Denote by $n=|U|+|V|$ and $m=|E|$ the total number of vertices and edges in $G$, respectively. Denote by $S=(U_S,V_S)$ a vertex subset of $G$, where $U_S\subseteq U$ and $V_S \subseteq V$. Let $v\in S$ indicate that $v\in U_S\cup V_S$. Let $G(S)=(U_S,V_S,E_S)$ be the subgraph of $G$ induced by $S$, where $E_S=\left\{(u,v)\mid u\in U_S\land v\in V_S\land (u,v)\in E\right\}$. For $u\in U$, denote by $N(u)=\left\{v\in V\mid (u,v)\in E\right\}$ and $\overline N(u)=\left\{v\in V\mid (u,v)\not\in E\right\}$ the set of neighbors and non-neighbors of $u$ in $G$, respectively. Similar definitions apply for the vertices in $U$. Let $d(u)=\lvert N(u)\rvert$ and $\overline d(u)=\lvert \overline N(u)\rvert$ be the degree and non-degree of $u$ in $G$, respectively. For a vertex subset (or a subgraph) $S$ of $G$, let $N_S(u)=N(u)\cap(U_S\cup V_S)$ and $\overline N_S(u)=\left\{v\in S\setminus N(u)\mid u\ \mathrm{and}\ v\ \mathrm{are\ on\ opposite\ sides}\right\}$ be the set of $u$'s neighbors and non-neighbors in $S$. respectively. Let $d_S(u)=\lvert N_S(u)\rvert$ and $\overline d_S(u)=\lvert\overline N_S(u)\rvert$ be the degree and non-degree of $u$ in $S$, respectively. For a bipartite graph $G=(U,V,E)$, we refer to $(u,v)\in U\times V$ as a non-edge of $G$ if $(u,v)\not \in E$. We denote $\overline E_S$ the set of non-edges in $G(S)$, i.e., $\overline E_S=U_S\times V_S\setminus E_S$. 
	The formal definition of $k$-defective biclique is given below.

	\begin{definition}[$k$-defective biclique]
		Given a bipartite graph $G=(U,V,E)$, a $k$-defective biclique $D=(U_{D},V_{D},E_{D})$ of $G$ is a subgraph of $G$ that has at most $k$ non-edges, i.e., $\lvert U_{D}\times V_{D}\setminus E_{D}\rvert\le k$.
	\end{definition}
	
	The $k$-defective biclique corresponds to the biclique if $k=0$, which indicates that $k$-defective biclique represents a more general structure. 
	If a $k$-defective biclique of $G$ contains the most number of edges among all $k$-defective bicliques in $G$, we refer to it as a maximum $k$-defective biclique in $G$. To facilitate the illustration, we denote the maximum $k$-defective biclique as $k$-MDB. 
	
	In this paper, we focus on the problem of searching for a $k$-MDB in a bipartite graph. However, we observe that a $k$-MDBs in many real-world graphs may exhibit a skewed structure, similar to that of maximum bicliques, where one side has a large number of vertices while the other side may contain only a few (or even one) vertices. Such $k$-MDBs often hold limited practical values in network analysis. To address this issue, we define a size threshold $\theta$, representing the minimum number of vertices required on both sides of the $k$-MDB. Furthermore, we impose the condition $\theta>k$ to ensure that the $k$-MDB is connected, as demonstrated in the following lemma.
	
	\begin{lemma}
		\label{lem:con}
		Given a bipartite graph $G=(U,V,E)$ and a $k$-defective biclique $D=(U_{D},V_{D},E_{D})$ of $G$, if $\lvert U_{D}\rvert>k$ and $\lvert V_{D}\rvert>k$, we establish that $D$ is a connected subgraph in $G$.
	\end{lemma}
	
	\begin{proofsketch}
		Assume that a $k$-defective biclique $D$ is not connected when $\lvert U_{D}\rvert>k$ and $\lvert V_{D}\rvert>k$. In this case, $D$ can be decomposed into two nonempty subgraphs $D_1=(U_{D_1},V_{D_1},E_{D_1})$ and $D_2=(U_{D_2},V_{D_2},E_{D_2})$, with the results of $E_{D_1} \cup E_{D_2} = E_D$.
		The number of non-edges between $D_1$ and $D_2$ is given by $|U_{D_1}|\cdot|V_{D_2}|+|U_{D_2}|\cdot|V_{D_1}|$, which exceeds $|U_D|$ (or $ |V_D|$) and is thus larger than $k$. This contradicts the assumption that $D$ is a $k$-defective biclique. Therefore, the proof is complete.
	\end{proofsketch}
	
	Based on Lemma \ref{lem:con}, we formally define the problem addressed in this paper as follows.
	
	\begin{problem}[MDB search problem]
		Given a bipartite graph $G=(U,V,E)$ and two integers $k$ and $\theta$ with $\theta>k$, the MDB search problem aims to find a connected $k$-MDB $D=(U_{D},V_{D},E_{D})$ of $G$ satisfying $\lvert U_{D}\rvert\ge\theta$ and $\lvert V_{D}\rvert \ge \theta$.
	\end{problem}
	
	\begin{theorem} \label{the:nph}
		The MDB search problem is NP-hard.
	\end{theorem} 
	\begin{proof}
	We establish the NP-hardness of the MDB search problem through a polynomial-time reduction from the classical NP-complete maximum clique search problem. Due to space limitation, we present the complete version of this proof and all the omitted proofs of this paper in \cite{RN0}.
	\end{proof}


		\section{Our Proposed Algorithms} \label{sec:alg}
	
	In this section, we propose two novel algorithms to solve the $k$-MDB search problem. Specifically, we first propose a new branch-and-bound algorithm that utilizes a newly-designed binary branching strategy, achieving the worst-case time complexity of $O(m\alpha_k^n)$, where $\alpha_k < 2$ (e.g. $k=1,2$, and $3$, we have $\alpha_k = 1.911, 1.979,$ and $1.995$, respectively). To further improve the efficiency, we propose a novel pivoting-based algorithm that incorporates a well-designed pivoting-based branching strategy with our proposed binary branching strategy. We prove that the time complexity of the improved algorithm is bounded by $O(m\beta_k^n)$, where $\beta_k < \alpha_k$ (e.g., when $k=1, 2, 3$, $\beta_k = 1.717, 1.856,$ and $1.931$, respectively). To the best of our knowledge, the algorithms we proposed in this section are the first two $k$-MDB search algorithms that surpass the trivial time complexity of $O^*(2^n)$.

	\subsection{A New Branch-and-bound Algorithm} \label{sec:mdbb}
	
	We firstly introduce a branch-and-bound algorithm for the $k$-MDB problem. It employs a straightforward branch-and-bound approach by selecting a vertex $v$ from the graph and dividing the $k$-MDB search problem into two subproblems: one seeks a $k$-MDB that includes $v$, while the other seeks a $k$-MDB that excludes $v$.
	Specifically, given a bipartite graph $G$, denote by $S=(U_S, V_S)$ a partial set and $C=(U_C, V_C)$ a candidate set, where $G(S)$ is a $k$-defective biclique of $G$, and for any vertex $u\in C$, $G(S\cup\{u\})$ forms a larger $k$-defective biclique of $G$. $(S, C)$ is referred to as an instance, which aims to identify a $k$-MDB consisting of $S$ and a subset of $C$. In particular, the instance $((\emptyset,\emptyset), (U,V))$ aims to find a $k$-MDB of $G$, while the instance $(S^*, (\emptyset,\emptyset))$ corresponds to a possible solution $G(S^*)$. To clarify, we refer to a vertex $u\in C$ as a \textit{branching vertex} when a new sub-instance is generated by expanding $S$ with $u$.
	
	Given an instance $(S,C)$, our branch-and-bound algorithm selects a branching vertex $u\in C$ and generates two new sub-instances: $(S\cup\{u\}, C\setminus\{u\})$ and $(S, C\setminus\{u\})$. The algorithm then recursively processes each sub-instance in the same manner until the candidate set becomes empty, where the partial set yields a possible solution. Finally, the $k$-MDB of $G$ is obtained as the yielded solution with most edges.
	The simplest approach to obtain a branching vertex is to randomly select one from $C$. However, this strategy may lead to the selection of all vertices in $C$, leading to a worst-case time complexity of $O^*(2^n)$ and too many unnecessary computations. To improve the efficiency and reduce the time complexity, we propose a new binary branching strategy. 
	
	Intuitively, if $k$ vertices from $C$ with non-neighbors in $S$ are successively added to $S$, then $S$ will contain at least $k$ non-edges. Consequently, this implies that all vertices in $C$ must be the neighbors of each vertex in $S$. Additionally, if no vertex in $C$ that has non-neighbors in $S$, adding a vertex $v$ from $C$ to $S$ will result in exactly $\overline d_C(v)$ vertices in $C$ with non-neighbors in $S$. Based on these observations, we give our new binary branching strategy as follows.
	
	\noindent\textbf{New binary branching strategy.} Given an instance $(S, C)$, let $u=\mathrm{\arg\max}_{u\in C}\overline d_S(u)$. The sub-instances are generated in the following way:
	\begin{itemize}
		\item If $\overline{d}_S(u)>0$, select $u$ as the branching vertex and generate two sub-instances: $(S\cup\{u\},C\setminus\{u\})$, $(S, C\setminus\{u\})$.
		\item Otherwise, select $v\in C$ with the maximum $\overline d_C(v)$ as the branching vertex, and generate two sub-instances: $(S\cup\{v\},C\setminus\{v\})$, $(S, C\setminus\{v\})$.
	\end{itemize}

	\begin{figure}[!t]
		\setlength{\belowcaptionskip}{-0.3cm} 
		\includegraphics[width=\columnwidth]{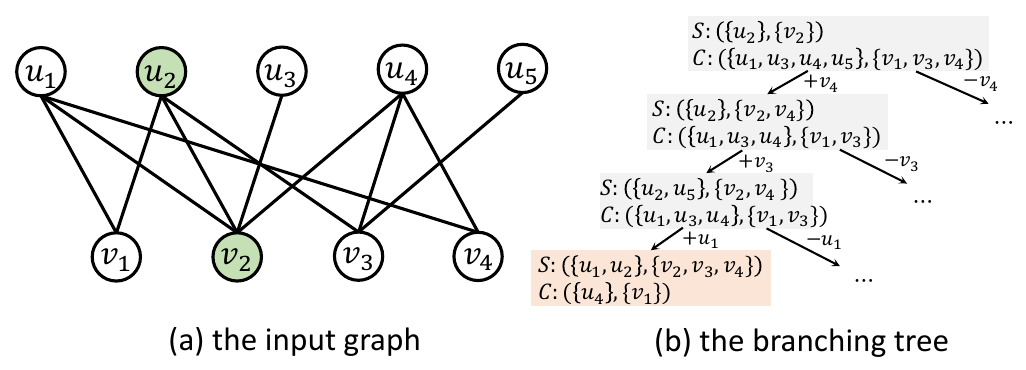}
		\caption{An example of the new binary branching strategy, where $k=2$ and current instance is $S=(\{u_2\},\{v_2\})$ (green vertices) and $C=(\{u_1,u_3,u_4,u_5\},\{v_1,v_3,v_4\})$ (uncolored vertices).}
		\label{fig:eg2}
	\end{figure}

	\begin{example}
	Fig. \ref{fig:eg2} shows an example of our new binary branching strategy for finding a $2$-MDB. The input graph is depicted in Fig. \ref{fig:eg2}a, where the current instance is $S=(\{u_2\},\{v_2\})$ (green vertices) and $C=(\{u_1,u_3,u_4,u_5\},\{v_1,v_3,v_4\})$ (uncolored vertices). Fig. \ref{fig:eg2}(b) illustrates the branching tree obtained by applying our binary branching strategy. Specifically, in the top-level branch, $v_4$ is selected as the branching vertex due to its maximum number of non-neighbors in $S$. When $v_4$ is added to $S$, it brings a non-edge in $S\cup\{v_4\}$. Then, we observe that there is no vertex in the second-level branch with non-neighbors in $S\cup\{v_4\}$. Therefore, $v_3$ is selected as the second-level branching vertex as it has the most non-neighbors in $C\setminus\{v_4\}$. In the third-level branch, $u_1$ is selected as the branching vertex as it has the most non-neighbors in $S\cup\{v_4,v_3\}$. In the forth-level, $S\cup\{u_1,v_3,v_4\}$ contains 2 non-edges, leading to the deletion of $u_3$ from $C$ since it still has non-neighbors in $S\cup\{u_1,v_3,v_4\}$, which is the key to reducing the unnecessary computations.
	\end{example}

	\begin{detail}
		Algorithm \ref{alg:mdbb} outlines the pseudocode of our proposed branch-and-bound algorithm. It takes a bipartite graph $G$ and two integers $k$ and $\theta$ as input parameters, and returns a $k$-MDB $D^*$ of $G$. It first initializes $D^*$ as an empty graph, and then starts to search $k$-MDB by invoking the procedure \brb where all vertices in $G$ serve as the candidate set (line 2). 
		\brb holds an instance $(S, C)$ as its parameter, where $S$ denotes the partial set and $C$ denotes the candidate set. Within \brb, the algorithm first determines whether a larger $k$-defective biclique has been found  (lines 5-8). After that, it selects a branching vertex $u$ from $C$ using the new binary branching strategy (lines 9-10) and generates two sub-instance $(S',C'\cup\{u\})$ and $(S, C\setminus\{u\})$ (lines 11-12). Here, $S'$ and $C'$ are the new partial set and new candidate set derived by invoking \upd (line 11), ensuring the candidate set keeps all valid branching vertex (lines 14-15). We next analyze the time complexity of Algorithm \ref{alg:mdbb}.
		
	\end{detail}
	
	\begin{algorithm}[!t]
		\caption{The branch-and-bound algorithm}
		\scriptsize
		\label{alg:mdbb}
		\KwIn{Bipartite graph $G=(U,V,E)$, integers $k$ and $\theta$}
		\KwOut{A $k$-MDB of $G$}
		
		$D^*\leftarrow$ an empty graph\;
		\brb{$(\emptyset, \emptyset)$, $(U, V)$}\;
		\Return {$D^*$}\;
		
		\Proc{\brb{$S=(U_S,V_S)$, $C=(U_C,V_C)$}} {
			\If{$C=\emptyset$} {
				\If{$\lvert E_{G(S)}\rvert >\lvert E_{D^*}\rvert$ \And $\lvert U_S\rvert\ge \theta$ \And $\lvert V_S\rvert\ge \theta$} {$D^*\leftarrow G(S)$\;}
				\Return{}\;
			}
			

			$u\leftarrow \mathrm{\arg\max}_{v\in C}\overline{d}_S(v)$\;
			
			\lIf{$\overline d_S(u)=0$} {
				$u\leftarrow \mathrm{\arg\max}_{v\in C}\overline{d}_C(v)$
			}
			
			$(S', C')\leftarrow$ \upd{$S$, $C$, $u$}\;
			\brb{$S'\cup \{u\}$, $C'$};\ 
			\brb{$S$, $C\setminus \{u\}$}\;
			
		}
		
		\Fn{\upd{$S$, $C$, $u$}} {
			$C'\leftarrow\{v\in C\setminus\{u\}\mid\overline{d}_{S\cup\{u\}}(v)\le k-\overline{d}(S)\}$\;
			$C'_0\leftarrow\{v\in C'\mid\overline{d}_{S\cup\{u\}\cup C'}(v)=0\}$\;
			\Return{$(S\cup C'_0, C'\setminus C'_0)$}\;
		}
		
	\end{algorithm}
	
	

	\begin{theorem}
		Given a bipartite graph $G=(U,V,E)$ and an integer $k$, the time complexity Algorithm \ref{alg:mdbb} is given by $O(m\alpha_k^n)$, where $\alpha_k$ is the largest real root of equation $x^{2k+5}-2x^{2k+4}+x^3-x^2+1=0$. Specifically, when $k=1,2$ and $3$, we have $\alpha_k=1.911,1.979$ and $1.995$, respectively.
	\end{theorem}

	\begin{proof}
		Let $T(n)$ denote the maximum number of leave instances generated by \brb (i.e. instances with an empty $C$), where $n=\lvert C\rvert=\lvert U_C\rvert+\lvert V_C\rvert$. Since each recursive call of \brb takes $O(m)$ time, the total time complexity of \brb is $O(m\cdot T(n))$. 
		
		Then, we analyze the size of $T(n)$ based on the binary branching strategy. 
		In each branch of \brb, either a vertex with at least one non-neighbor in $S$ or a vertex with the most non-neighbors in $C$ is selected as the branching vertex. If the vertex $u$ with at least one non-neighbor in $S$ is selected, it results in at least one additional non-edge in $S\cup \{u\}$. Alternatively, if a vertex $v$ with the most non-neighbors in $C$ is selected as the branching vertex, the next-level branch will contain at least $\overline{d}_C(v)$ vertices in $C$ that have non-neighbors in $S\cup \{v\}$, where $\overline{d}_C(v) > 0$. Consequently, we can conclude that in the worst case, adding two vertices from $C$ to $S$ will result in at least one additional non-edge in $S$. After adding at most $2k+1$ vertices from $C$ to $S$, we have $\lvert \overline E_S\rvert =k$ and there exists a vertex $v\in C$ such that $\overline d_S(v)>0$. Based on this, the recurrence is given by $T(n)\le \sum_{i=1}^{2k+1} T(n-i)+T(n-2k-2)$.

			
		Note that the $(2k+2)$-th level of branching (the term $T(n-2k-2)$) corresponds to a branch that searches for the maximum biclique. In this case, it is evident that any vertex in $C$ selected as the branching vertex will result in at least one additional vertex being removed from $C$. Thus, we yield a recurrence of $T(n-2k-2)=T(n-2k-3)+T(n-2k-4)$ for such a branch. Based on above discussion, we can present the following tighter recurrence.
		\begin{align}
			T(n)\le \sum_{i=1}^{2k+1} T(n-i)+T(n-2k-3)+T(n-2k-4).
		\end{align}
		According to the theory proposed in \cite{RN76}, $T(n)$ can be expressed in the form of $\alpha_k^n$, where the value of $\alpha_k$ derives from the largest real root of equation $x^n=x^{n-1}+x^{n-2}+\cdots+x^{n-2k-1}+x^{n-2k-3}+x^{n-2k-4}$, which can be simplified to $x^{2k+5}-2x^{2k+4}+x^3-x^2+1=0$. Therefore, the theorem is proven.
	\end{proof}
	
	\begin{lemma}
		Given a bipartite graph $G=(U,V,E)$, the time complexity of Algorithm \ref{alg:mdbb} for finding a $0$-MDB of $G$ is $O(m\times 1.618^n)$.
	\end{lemma}
	\begin{proofsketch}
		When $k=0$, the vertex $v\in C$ with maximum $\overline d_C(v)$ is selected as the branching vertex. Moving $v$ from $C$ to $S$ results in the deletion of $v$'s non-neighbors in $C$, leading to the recurrence $T(n)\le T(n-1)+T(n-2)$, which yields $T(n)=1.618^n$.
	\end{proofsketch}
	
	\subsection{A Novel Pivoting-based Algorithm} \label{sec:mdbp}
	
	We observe that Algorithm \ref{alg:mdbb} may still generate numerous unnecessary sub-instances. Consider an instance $(S, C)$ and a branching vertex $v\in C$ with $\overline d_S(v)=0$. In Algorithm \ref{alg:mdbb}, if $v$ and all non-neighbors of $v$ have been excluded from the current instance, the remaining vertices cannot form any $k$-MDB. This is because any $k$-defective biclique derived from these remaining vertices can always be expanded by including $v$. Consecutively, Algorithm \ref{alg:mdbb} may produce many $k$-defective bicliques that are not maximum. To address the problem and further improve the efficiency, we propose a novel pivoting-based algorithm in this section. We formalize above insights into the following lemma.
	
	\begin{lemma}
		\label{lem:piv}
		Given an instance $(S,C)$ and any vertex $v\in C$ with $\overline d_S(v)=0$, the $k$-MDB of instance $(S,C)$ must contain either $v$ or a vertex in $\overline N_C(v)$.
	\end{lemma}
	

	Based on Lemma \ref{lem:piv}, we develop a pivoting strategy that provides a new way to generate sub-instances. The details are presented as follows. 
	
	\noindent\textbf{Pivoting strategy.} Consider an instance $(S, C)$ and a vertex $u\in C$ with $\overline d_S(u)=0$. Let $u$ be the \textit{pivot vertex}. Assume $\overline N_C(u)=\{v_1,v_2,\dots,v_{r}\}$, where $r=\overline d_C(u)$. The pivoting strategy treats each vertex in $\{u\}\cup \overline N_C(u)$ as a branching vertex and generates exactly $r+1$ sub-instances. The first sub-instance is given by $(S\cup\{u\}, C\setminus\{u\})$, and for $i\ge 2$, the $i$-th sub-instance is given by $(S\cup\{v_{i-1}\}, C\setminus\{u, v_1,v_2,\dots,v_{i-1}\})$. 
	
	With this strategy, $S$ is expanded only with the vertex $u$ and vertices in $\overline N_C(u)$, avoiding redundant sub-instances that would lead to non-maximum $k$-defective biclique. Next, we consider the selection of the pivot vertex. Let $C_0\subseteq C$ be the set of vertices fully connected to $S$, i.e., $C_0=\{v\in C\mid\overline d_S(v)=0\}$. According to the pivoting strategy, any vertex from $C_0$ can serve as a pivot vertex. To minimize the number of sub-instances, the vertex $u\in C_0$ with the fewest non-neighbors in $C$ is chosen as the pivot vertex. 
	
	In the case where $C_0$ is empty, every vertex in $C$ has at least one non-neighbor in $S$, making the pivoting strategy inapplicable. Fortunately, as discussed in Sec. \ref{sec:mdbb}, continuously branching $k$ times by adding vertices from $C$ to $S$ will result in $S$ containing at least $k$ non-edges. Moreover, $C$ becomes an empty set after this process, which significantly reduces unnecessary computations. Motivated by this observation, we integrate the pivoting strategy with the binary branching strategy to develop a novel pivoting-based branching strategy as outlined below.

	\noindent\textbf{Novel pivoting-based branching strategy.} Given an instance $(S, C)$, let $C_0=\{v\in C\mid\overline d_S(v)=0\}$. The new sub-instances of $(S, C)$ are generated in the following way:
	\begin{itemize}
		\item If $C_0$ is empty or $\lvert C\setminus C_0\rvert>k-\lvert\overline E_S\rvert$, let $u=\mathrm{\arg\max}_{v\in C} \overline{d}_S(v)$ be the branching vertex and generate two sub-instances: $(S\cup\{u\}, C\setminus\{u\})$ and $(S, C\setminus\{u\})$.
		\item Otherwise, let $u=\mathrm{\arg\min}_{v\in C_0}\overline{d}_C(v)$, and then:
		\begin{itemize}
			\item If $\overline d_C(u) > k-\lvert\overline E_S\rvert > 0$, generate two sub-instances with $u$ as the branching vertex: $(S\cup\{u\}, C\setminus\{u\})$ and $(S, C\setminus\{u\})$.
			\item Otherwise, assume $\overline N_C(u)=\{v_1,v_2,\dots,v_{r}\}$, where $r=\overline d_C(v)$, and generate $\overline d_C(u)+1$ sub-instances with $u$ as the pivot vertex, where the first sub-instance is given by $(S\cup\{u\}, C\setminus\{u\})$, and for $i\ge 2$, the $i$-th sub-instance is given by $(S\cup\{v_{i-1}\}, C\setminus\{u, v_1,v_2,\dots,v_{i-1}\})$.
		\end{itemize}  
	\end{itemize}

	\begin{figure}[!t]
		\centering
		\setlength{\belowcaptionskip}{-0.3cm} 
		\includegraphics[width=\columnwidth]{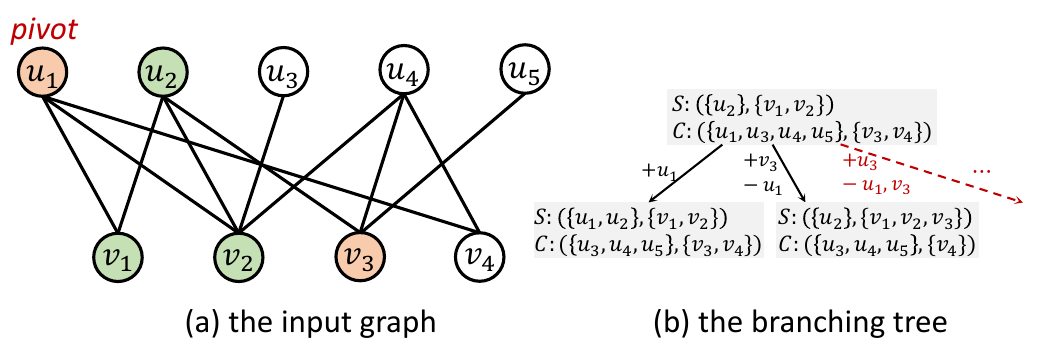}
		\caption{An example of the pivoting-based branching strategy, where $k=2$ and current instance is set as $S=(\{u_2\},\{v_2\})$ (green vertices) and $C=(\{u_1,u_3,u_4,u_5\},\{v_1,v_3,v_4\})$ (orange and uncolored vertices).}
		\label{fig:eg3}
	\end{figure}
	
	\begin{example}
		Fig. \ref{fig:eg3} shows an example of the our pivoting-based branching strategy for the $2$-MDB search. The example graph is given by Fig. \ref{fig:eg3}a, with the current instance set as $S=(\{u_2\},\{v_1,v_2\})$ (green vertices) and $C=(\{u_1,u_3,u_4,u_5\},\{v_3,v_4\})$ (orange and uncolored vertices). According to the pivoting-based branching strategy, $C_0=(\{u_1\},\{v_3\})$, and $u_1$ is selected as the pivot vertex since it has the fewest non-neighbors in $C$ among vertices in $C_0$. Then, only $u_1$ and all non-neighbors of $u_1$ in $C$ are used to generate new sub-instances (orange vertices), while other vertices of $C$ are not considered (uncolored vertices). This results in only two sub-instances generated by our pivot-based branching strategy. The branching tree for this example is illustrated in Fig. \ref{fig:eg3}b, where the black arrows represent the generated branches while the red arrows indicate the unnecessary branches.
	\end{example}

		\begin{detail}
			Algorithm \ref{alg:mdbp} outlines the pseudocode of the proposed pivot-based algorithm. It begins by searching for $k$-MDB by invoking the procedure \brp with an empty partial set $S$ and the entire vertex set of $G$ as the candidate set $C$ (line 2). The algorithm returns the global variable $D^*$ as the final $k$-MDB after \brp is processed (line 3). Within \brp, the algorithm first updates $D^*$ if a larger $k$-defective biclique is found (lines 5-7). Then, it generates new sub-instances based on the pivoting-based branching strategy (lines 8-22). Specifically, if $C_0=\emptyset$ or $\lvert C\setminus C_0\rvert>k-\lvert\overline E_S\rvert$, the vertex $u= \mathrm{\arg\max}_{v\in C}\overline{d}_S(v)$ is selected as the branching vertex (line 11), and two new sub-instances are generated (line 13). Otherwise, the algorithm assigns $u=\mathrm{\arg\min}_{v\in C_0}\overline{d}_C(v)$. If $\overline d_C(u)>k-\lvert \overline E_S\rvert$, two sub-instances are generated with $u$ as the branching vertex (lines 16-18). Otherwise, $\overline d_C(u)+1$ new sub-instances are generated with $u$ and all $u$'s non-neighbors in $C$ as branching vertex (lines 20-22). When a branching vertex is added to $S$, the algorithm adopts the \upd procedure presented in Algorithm \ref{alg:mdbb} to update the current instance (lines 12, 17 and 21). Below, we analyze the time complexity of Algorithm \ref{alg:mdbp}.
		\end{detail}
		
			\begin{algorithm}[!t]
			\caption{The pivoting-based algorithm}
			\scriptsize
			\label{alg:mdbp}
			\KwIn{Bipartite graph $G=(U,V,E)$, integers $k$ and $\theta$}
			\KwOut{A $k$-MDB of $G$}
			
			$D^*\leftarrow$ an empty graph\;
			\brp{$(\emptyset, \emptyset)$, $(U, V)$}\;
			\Return {$D^*$}\;
			
			\Proc{\brp{$S=(U_S,V_S)$, $C=(U_C,V_C)$}} {
				\If{$C=\emptyset$} {
					\If{$\lvert E_{G(S)}\rvert >\lvert E_{D^*}\rvert$ \And $\lvert U_S\rvert\ge \theta$ \And $\lvert V_S\rvert\ge \theta$} {$D^*\leftarrow G(S)$\;}
					\Return{}\;
				}
				
				$C_0\leftarrow \{v\in C\mid\overline d_{S}(v)= 0\}$\;
				

					\If{$C_0= \emptyset$ \Or $\lvert C\setminus C_0\rvert >k-\lvert \overline E_S\rvert$} {
						$u\leftarrow \mathrm{\arg\max}_{v\in C}\overline{d}_S(v)$\;
						$(S', C')\leftarrow$ \upd{$S$, $C$, $u$}\;
						\brp{$S'\cup \{u\}$, $C'$};\ \brp{$S$, $C\setminus \{u\}$}\;
					}
					\Else {
						$u\leftarrow \mathrm{\arg\min}_{v\in C_0}\overline{d}_C(v)$\;
						\If {$\overline d_C(u)>k-\lvert \overline E_S\rvert>0$} {
							$(S', C')\leftarrow$ \upd{$S$, $C$, $u$}\;
							\brp{$S'\cup \{u\}$, $C'$};\ \brp{$S$, $C\setminus \{u\}$}\;
						}
						\Else {
							\For{$v\in \{u\}\cup\overline{N}_C(u)$}{
								$(S', C')\leftarrow$ \upd{$S$, $C$, $v$}\;
								\brp{$S'\cup \{v\}$, $C'$};\ $C\leftarrow C\setminus \{v\}$\;
							}
						}
					}
					
				}
				
			\end{algorithm}
		

		\begin{theorem}
			\label{thm:mdbp}
			Given a bipartite graph $G=(U,V,E)$ and an integer $k$, the time complexity of Algorithm \ref{alg:mdbp} is given by $O(m\beta_k^n)$, where $\beta_k$ is the largest real root of equation $x^{2k+5}-2x^{2k+4}+x^{k+3}-2x+2=0$. Specifically, when $k=1,2$ and $3$, we have $\beta_k=1.717, 1.856$ and $1.931$, respectively.
		\end{theorem}
		
		\begin{proof}
			In Algorithm \ref{alg:mdbp}, Let $T(n)$ denote the maximum number of leave instances generated by \brp, where $n=\lvert C\rvert =|U_C|+|V_C|$. It is easy to verify that finding the vertex with most (least) non-neighbors in $S$ ($C$) takes at most $O(m)$ time. Then, each recursive call of \brp takes $O(m)$ time, and thus \brp takes $O(m\cdot T(n))$ time.
			We next focus on analyzing the size of $T(n)$. Based on the pivoting-based branching strategy, there are three cases that need to be discussed:
			\begin{enumerate}
				\item If $C_0=\emptyset$ or $\lvert C\setminus C_0\rvert > k-\lvert\overline E(S)\rvert$, \brp continuously applies the binary branching strategy. After moving at most $k$ vertices from $C$ to $S$, $S$ contains $k$ non-edges. At this point, at least one vertex in $C$ will be removed as it has non-neighbors in $S$. Therefore, the recurrence for this case is given by
				$T(n)\le \sum_{i=1}^{k+1} T(n-i)$.
				\item If $k-\lvert\overline E_S\rvert=0$, \brp degenerates into maximum biclique search and continuously applies the pivoting strategy. Let $r=\overline d_C(u)$, the recurrence for this case is given by $T(n)\le T(n-r-1)+\sum_{i=1}^{r}T(n-r-i)$. If $r\le1$, we have $T(n)\le 2T(n-2)$. If $r\ge 2$, let $T(n)=\gamma_r^n$ according to the existing theory \cite{RN76}, where $\gamma_r$ is the largest real root of equation $x^n=x^{n-r-1}+\sum_{i=1}^{r} x^{n-r-i}$. The value of $\gamma_r$ is monotonically decreasing with respect to $r$ when $r\ge 2$. By solving the equation we have $\gamma_2=1.395<\sqrt{2}$. As a result, we can always say $T(n)\le 2T(n-2)$ for this case.
				

				\item Otherwise, let $u=\mathrm{\arg\min}_{v\in C_0}\overline{d}_C(v)$.  
				If $\overline d_C(u)\le k-\lvert\overline E(S)\rvert$, \brp performs the pivoting strategy with $u$ as the pivot vertex. This strategy generates at most $ \overline{d}_C(u)+1 \le k+1$ sub-instances. Therefore, the recurrence for this case is given by $T(n)\le \sum_{i=1}^{k+1} T(n-i)$.
				Otherwise, \brp uses $u$ to perform binary branching strategy. In the next-level branching, we have $|C\setminus C_0| > k- \lvert\overline E(S \cup \{u\})\rvert$, as $\overline d_C(u) > k-\lvert\overline E(S)\rvert$. In this case, the sub-branch (the term $T(n-1)$) will continuously perform the binary branching strategy as discussed in case (1). Consequently, the recurrence for this case is given by $T(n)\le\sum_{i=1}^{k+1} T(n-i)+T(n-2k-2)$. Note that the term $T(n-2k-2)$ degenerates to maximum biclique search, and we can derive that $T(n-2k-2) = 2 T(n-2k-4)$ as discussed in case (2). Therefore, the final recurrence for this case is given by 
					$T(n)\le \sum_{i=1}^{k+1}T(n-1)+2T(n-2k-4)$. 
			\end{enumerate}
			\hspace{1em} By combining the above analysis, the worst-case recurrence is given by
			\begin{align} 
				\setlength\abovedisplayskip{0pt}
				\setlength\belowdisplayskip{0pt}
			T(n)\le \sum_{i=1}^{k+1}T(n-1)+2T(n-2k-4). 
			\end{align}
			
			Based on the existing theorem proposed in \cite{RN76}, we denote by $T(n)=\beta_k^n$, where $\beta_k$ is the largest real root of equation $x^n=x^{n-1}+x^{n-2}+\cdots+x^{n-k-1}+2x^{n-2k-4}$. The equation is equivalent to $x^{2k+5}-2x^{2k+4}+x^{k+3}-2x+2=0$. Consequently, the theorem is proven.
		\end{proof}
		\begin{lemma}
			Given a bipartite graph $G=(U,V,E)$, the time complexity of Algorithm \ref{alg:mdbp} for finding a $0$-MDB of $G$ is $O(m\cdot1.414^n)$. 
		\end{lemma}	
		\begin{proofsketch}
			When setting $k=0$, Algorithm \ref{alg:mdbp} will always execute the pivoting-based branching strategy (lines 20-22), corresponding to a recurrence of $T(n) < 2T(n-2)$ according to the proof of Theorem \ref{thm:mdbp}. Therefore, we have $T(n)=1.414^n$.
		\end{proofsketch}

		\noindent\textbf{Discussion.}
		It is worth noting that the branch-and-bound and pivoting techniques has been widely applied on many cohesive subgraph models such as the $k$-defective clique \cite{RN2, RN3, RN4}, the $k$-biplex \cite{RN71, RN72} and the biclique \cite{RN45, RN46, RN44}. However, our binary branching strategy and pivoting-based branching strategy are significantly different from all these existing solutions. 
		
		Compared to the state-of-the-art maximum $k$-defective clique algorithm in \cite{RN2}, which achieves a non-trivial time complexity by utilizes a condition to restrict the scenarios of pivoting, directly applying this approach to our pivoting-based algorithm would degrade the time complexity to $O^*(2^n)$. In contrast, we employ a more refined pivoting strategy that considers both the number of non-edges between $S$ and $C$ (lines 10-13 in Algorithm \ref{alg:mdbp}) and the number of pivot vertex's non-neighbors (lines 15-18 in Algorithm \ref{alg:mdbp}), thereby effectively bounding the time complexity. 
		
		Compared to the state-of-the-art pivoting-based maximum biplex search algorithm in \cite{RN71}, which achieves a time complexity of $O(m\gamma_k^n)$, with $\gamma_k=1.754$ when $k=1$, our algorithm \mdbp reaches a better time complexity. This is mainly because the algorithm in \cite{RN71} only considers prioritizing $k$ non-neighbors of $S$, which is similar to our binary branching strategy in Sec. \ref{sec:mdbb}. By contrast, our pivoting-based branching strategy also bounds the case that fewer than $k$ non-neighbors are between $S$ and $C$, which reduces the size of branching tree and thereby improves the time complexity (see the proof of Theorem \ref{thm:mdbp}).
		
		Compared to the maximal biclique algorithm in \cite{RN45}, our pivot vertex selection approach is much different. During the pivoting process, \cite{RN45} considers the common neighbors, and selects pivot vertex in a heuristic topological order from only one side, which may not be always optimal. In contrast, our pivoting-based branching strategy selects pivot vertex from both side through deterministic condition, which always generates fewest new sub-branches in each recursion.
	
	\section{Optimization Techniques} \label{sec:ot}
	
	In this section, we propose various optimization techniques, including graph reduction techniques to eliminate redundant vertices, upper-bounding techniques to prune unnecessary instances, and a heuristic approach for identifying an initial large $k$-defective biclique. These techniques can be easily integrated into our proposed algorithms in Sec. \ref{sec:alg}, and can significantly enhance their efficiency.
	
	\subsection{Graph Reduction Techniques} \label{sec:red}

	\noindent\textbf{Common-neighbor-based reduction.} For two vertices $u$ and $v$ on the same side, let $cn(u,v)$ (or $cn_S(u,v)$) denote the number of common neighbors of $u$ and $v$ in $G$ (or in $S$). 
	Given a bipartite graph $G=(U,V,E)$ and a $k$-MDB $D=(U_D,V_D,E_D)$ of $G$, every edge $(u, v)\in E_D$ satisfies that (1) $d_{D}(u)\ge\theta-k$, and (2) for all $w\in N_{D}(v)$, $cn_D(u,w)\ge\theta-k$. By combining the two conditions, we can conclude that for an edge $(u,v)\in E$, if $v$ has fewer than $\theta-k$ neighbors $w$ such that $cn(u,w)\ge\theta-k$, the edge $(u,v)$ can be pruned in $G$. 
	
	Based on this idea, we present Algorithm \ref{alg:cnred} to prune such edges from $G$, referred to as \cnred. The algorithm enumerates all edges $(u, v)$ in ascending order of $u$'s degree (lines 1-2). For each edges $(u,v)$, the algorithm counts common neighbors between $u$ and the neighbors $w$ of $v$ (lines 3-8). If $v$ has fewer than $\theta-k$ neighbors $w$ satisfying $cn(u,w)\ge\theta-k$, the edge $(u,v)$ is pruned from $G$ (lines 9-11). Additionally, if the degree of $u$ or $v$ drop below $\theta-k$ due to the deletion of $(u,v)$, the algorithm will also prune $u$ or $v$ from $G$ (lines 12-13). We give the time complexity of Algorithm \ref{alg:cnred} in the following lemma.
	\begin{lemma}
		Given a bipartite graph $G=(U,V,E)$, the complexity of Algorithm \ref{alg:cnred} is $O(d\cdot m)$, where $d=\max_{u\in U\cup V}{d(u)}$.
	\end{lemma}
	In practice, Algorithm \ref{alg:cnred} performs efficiently although its time complexity is $O(d\cdot m)$. This is because in real-world graphs, most vertices have degrees much smaller than $d$. As a result, the actual runtime of the algorithm often does not reach this worst-case time complexity.
	 
	\begin{algorithm}[!t]
		\caption{\protect \cnred{$G$, $k$, $\theta$}}
		\scriptsize
		\KwIn{Bipartite graph $G=(U,V,E)$, integers $k$ and $\theta$}
		\KwOut{A reduced bipartite graph}
		\label{alg:cnred}
		$\mathcal{O}\leftarrow$ the ascending degree order of $U\cup V$\;
		\For{$u\in \mathcal{O}$} {
			\For{$v\in N(u)$} {
				\For{$w\in N(v)$} {
					$cn(w)\leftarrow 0$\;
				}
			}
			\For{$v\in N(u)$} {
				\For{$w\in N(v)$} {
					$cn(w)\leftarrow cn(w)+1$\;
				}
			}
			\For{$v\in N(u)$} {
				$N'\leftarrow \left\{w\in N(v)\mid cn(w)\ge\theta-k\right\}$\;
				
				\lIf{$\lvert N'\rvert < \theta-k$} {
					$E\leftarrow E\setminus (u,v)$
				}
			}
			\lIf{$d(u) < \theta-k$} {$U\leftarrow U\setminus u$}
			\lFor{$v\in N(u)$ s.t. $d(v)<\theta-k$} {
				$V\leftarrow V\setminus v$
			}
		}
		
		\Return{$G$}\;
		
	\end{algorithm}

	
	\noindent\textbf{One-non-neighbor reduction.} 
	For an instance $(S, C)$ and a vertex $u$, let $C_u$ denote the subset of $C$ where all vertices are on the same side as $u$. If a branching vertex has only one non-neighbor in $S\cup C$, removing it from $C$ can lead to the pruning of additional vertices. The following lemma explains this idea in detail.
	
	\begin{theorem}
		\label{thm:onn}
		Given an instance $(S,C)$ and a vertex $u\in C$ with $\overline d_{S\cup C}(u)=1$, if $u$ is removed from $C$, all vertices $v\in C_u$ with $\overline d_S(v)\ge 1$ can be pruned from $C$. Moreover, if $d_C(u)=1$ and $w$ is the unique non-neighbor of $u$, then all vertices in $\overline N_C(w)$ can also be pruned from $C$.
	\end{theorem}
	
	\begin{proof}
		Assume $D$ is a $k$-MDB derived from $(S,C)$ excluding $u$. For any $v\in C_u$ such that $v\in D$ and $\overline d_S(v)\ge 1$, we have $D\setminus \{v\}\cup\{u\}$ is also a $k$-MDB. In the case where $\overline d_C(u)=1$ and $w$ is the unique non-neighbor of $u$, if $w\not\in D$, then $D\cup\{u\}$ forms a larger $k$-MDB. Otherwise, for any $v\in \overline N_C(w)$ such that $v\in D$, we have $D\setminus\{v\}\cup\{u\}$ is also a $k$-MDB. 
	\end{proof}
	
	\noindent\textbf{Ordering-based reduction.}
	Given a bipartite graph $G=(U,V,E)$, let $u_1,u_2,\dots,u_{\lvert U\rvert}$ be an arbitrary vertex order of $U$. For each vertex $u_i$, the ordering-based reduction aims to find a $k$-MDB $D_i$ including $u_i$ while excluding $u_1,u_2,\dots,u_{i-1}$. Obviously, the $k$-MDB of $G$ is the one with most edges among $D_1,D_2,\dots,D_{\lvert U\rvert}$. Next, we extend Lemma \ref{lem:con} to constrain the vertex distance in $D_i$, as shown in the following lemma.
	
	\begin{lemma}
		\label{lem:dis}
		Given a $k$-defective biclique $D=(U_D,V_D,E_D)$, if $\lvert U_{D}\rvert>k$ and $\lvert V_{D}\rvert>k$, the maximum distance between any two vertices of $D$ is 3.
	\end{lemma}

	Lemma \ref{lem:dis} indicates that when searching for $D_i$, vertices with distances more than $3$ from $u_i$ can be pruned. Let $N^2(u)=\{w\in N(v)\mid v\in N(u)\}$, i.e., the set of $u$'s 2-hop neighbors. Let $N_{+}^2(u_i)=N^2(u_i)\cap\{u_{i+1}, u_{i+2}, u_{\lvert U\rvert}\}$. According to Lemma \ref{lem:dis}, we can construct an initial instance $(S_i, C_i)$ to find $D_i$, where $S_i=(\{u_i\},\ \emptyset)$ and $C_i=(N_{+}^2(u_i),\ \bigcup_{u_j\in N_{+}^2(u_i)}N(u_j))$. 
	
	In practical implementation, to identify a larger $k$-defective biclique as early as possible, we traverse the vertices in $U$ in descending order of their degrees. This approach is based on the intuition that vertices with higher degrees are more likely to be present in an MDB. 
	Additionally, the common-neighbor-based reduction techniques can be employed on $(S_i, C_i)$ to prune unnecessary vertices and edges. Specifically, vertices $u\in C_i$ are pruned if $d_{S_i\cup C_i}(u)<\theta-k$, and edges $(u,v)$ are pruned if there are fewer than $\theta-k$ neighbors $w$ of $v$ satisfying $cn(u,w)\ge\theta-k$.
	
	
	\noindent\textbf{Progressive bounding reduction.} 
	The size threshold $\theta$ is provided as an input parameter and can be set to a very small value (e.g., $=k+1$), which may limit the effectiveness of our optimization techniques. To address this issue, we utilize the progressive bounding technique inspired from the maximum biclique problem \cite{RN13} to provide tighter size thresholds for the $k$-MDB search problem. 
	
	Specifically, let $\theta_U^t$ and $\theta_V^t$ denote the size thresholds in $t$-th round. The objective of $t$-th round is to identify a $k$-MDB $D_t=(U_{D_t},V_{D_t},E_{D_t})$ satisfying $\lvert U_{D_t} \rvert\ge \theta_U^t$ and $\lvert V_{D_t} \rvert\ge \theta_V^t$. Let $D^*=\max\{D_1,D_2,\dots,D_{t-1}\}$ be the currently largest $k$-defective biclique. The progressive bounding reduction generates $\theta_U^t$ and $\theta_V^t$ in the following way:
	\begin{align}
		\theta_U^t &= \max\{\theta, \lfloor\theta_U^{t-1}/2 \rfloor\}, \label{eq:6} \\
		\theta_V^t &= \max\{\theta,\lfloor\lvert E_{D^*}\rvert/\theta_U^{t-1}\rfloor\}. \label{eq:7}
	\end{align}
	
	Specially, $\theta_U^0=\max_{u\in U}d_U(u)+k$, and the iteration terminates when $\theta_U^t=0$. The following lemma shows the correctness of the progressive bound reduction.
	\begin{lemma}
		\label{lem:pb}
		Given a bipartite graph $G$ and a $k$-MDB $D=(U_D,V_D,E_D)$ of $G$, there always exists an integer $t$ such that $\theta_U^t\le\lvert U_{D}\rvert$ and $\theta_V^t\le\lvert V_{D}\rvert$.
	\end{lemma}
	
	In many real-world graphs, $\theta_U^t$ and $\theta_V^t$ are often significantly larger than $\theta$, which helps facilitate more effective graph reduction, such as a smaller reduced graph by performing the common-neighbor-based reduction.
	
	\noindent \textbf{Discussion.}
		Several prior studies have proposed similar graph reduction techniques for other subgraph models \cite{RN2, RN71, RN75, RN102, RN105}. The progressive bounding reduction have been widely applied on these models \cite{RN71,RN102,RN105}. In this work, we adapted these techniques to our problem by incorporating the unique properties of $k$-defective bicliques. However, our proposed common-neighbor-based reduction and one-non-neighbor reduction are technically distinct from all prior approaches. 
		
		On the one hand, the common-neighbor-based reduction leverages the constraints on vertex degree and common neighbors in a $k$-MDB. The reduction based on vertex degree has been widely applied in the form of core \cite{RN71, RN105}. However, this approach results in a very limited reduction in graph size. Additionally, by applying a single constraint on common neighbors of two vertices $u$ and $v$ where $cn(u,v)<\theta-k$, we can conclude that $u$ and $v$ cannot coexist in the same $k$-MDB. However, maintaining such coexistence relationships requires $O(n^2)$ space, making it quite challenging for practical implementation. 
		
		On the other hand, the one-non-neighbor reduction focuses on analyzing vertices in $C$ that have only one non-neighbor in the current subgraph $G(S\cup C)$. A relevant prior technique addresses the maximum $k$-defective clique search problem \cite{RN2}, where a vertex in $C$ with exactly one non-neighbor in the subgraph can be directly added to $S$. However, this reduction cannot be directly applied to the $k$-MDB problem, as such a vertex may not belong to any $k$-MDB. For example, consider the following $1$-MDB search instance $(S,C)$, where $\lvert U_S\rvert=2$, $\lvert V_S\rvert=3$, $U_C=\{u\}$ and $V_C=\{v\}$. $G(S)$ has one non-edge, $u$ and $v$ are non-neighbors. Directly adding $v$ to $S$ results in a non-maximum solution, because $G(S\cup\{v\})$ has $7$ edges while adding $u$ to $S$ yields a $k$-defective biclique with $8$ edges.
		In contrast to \cite{RN2}, we focus on the deletion of such vertices rather than their addition, which restricts the pruning scope to vertices on the same side (as discussed in Theorem \ref{thm:onn}), thereby ensuring the correctness of the searching result.

	\subsection{Novel Upper-Bounding Techniques} \label{sec:ub}

	In this subsection, we propose a novel upper-bounding technique that provides effective upper bounds for both vertices and edges. 
	Formally, given an instance $(S,C)$, assume $u_1,u_2,\cdots,u_{|U_C|}$ and $v_1,v_2,\cdots,v_{|V_C|}$ are the vertices of $U_C$ and $V_C$ in ascending order of $\overline d_S(\cdot)$, respectively. 
	Let $\overline d_S^i(U_C)=\sum_{t=1}^{i}\overline d_S(u_i)$ and $\overline d_S^i(V_C)=\sum_{t=1}^{i}\overline d_S(v_i)$, i.e., the total number of non-neighbors in $S$ among the first $i$ vertices of $U_C$ and $V_C$. Based on this, the vertex upper bound is presented in the following lemma.

	\begin{lemma}[Vertex upper bounds]
		\label{lem:ivub}
		For any $k$-MDB $D=(U_{D},V_{D},E_{D})$ derived from $(S,C)$, we have $\lvert U_{D}\rvert\le |U_S|+i$ and $\lvert V_{D}\rvert\le |V_S|+j$, where $i$ and $j$ are the largest values satisfying $\overline d_S^i(U_C)\le k-\lvert\overline E_S\rvert$ and $\overline d_S^j(V_C) \le k-\lvert\overline E_S\rvert$, respectively.
	\end{lemma}
	
	\begin{proofsketch}
		For the side of $U$, it is obvious that $i$ is the maximum number of vertices that can be added from $U_C$ to $U_S$ such that $(U_S\cup \{u_1,u_2,\dots,u_i\}, V_S)$ remains a valid $k$-defective biclique. For any $k$-MDB derived from $(S,C)$ consisting of $(U_S\cup U_{C'}, V_S\cup V_{C'})$, $(U_S\cup U_{C'}, V_S)$ is still a $k$-defective biclique. Therefore, we have that $|U_{C'}|\le i$. The proof on the other side is similar.
	\end{proofsketch}
	
	Based on Lemma \ref{lem:ivub}, if both sides reach their vertex upper bounds, an edge upper bound can be calculated as $(\lvert U_S\rvert+i)\times (\lvert V_S\rvert+j)-\lvert\overline E_S\rvert-\overline d_S^i(U_C)-\overline d_S^j(V_C)$. However, this upper bound does not account for the coexistence of vertices in $U_C$ and $V_C$. To solve the issue, we present a tighter edge upper bound in the following lemma.
	
	\begin{lemma}[Edge upper bound]
		\label{lem:ieub}
		For any $k$-MDB $D$ derived from $(S,C)$, we have that 
		$
		\lvert E_{D}\rvert\le \max_{0\le i\le\lvert U_C\rvert}(\lvert U_S\rvert+i)\times(\lvert V_S\rvert+j)-\lvert\overline E_S\rvert-\overline d_S^i(U_C)-\overline d_S^j(V_C) 
		$, where $j$ is the largest value satisfying $\overline d_S^i(U_C)+\overline d_S^j(V_C)\le k-\lvert\overline E_S\rvert$.
	\end{lemma}
	
	\begin{proofsketch}
		Consider a new instance $(S,C')$ reconstructed from $(S,C)$ by fully connecting vertices between $U_{C'}$ and $V_{C'}$, i.e., $\overline E_C=\emptyset$. The ordering of vertices in $C'$ with respected to $\overline d_S(\cdot)$ remains unchanged. There exists a $k$-MDB derived from $(S,C')$ composed of $S$, a prefix of $\{u_1,u_2,\dots,u_{\lvert U_C\rvert}\}$, and a prefix of $\{v_1,v_2,\dots,v_{\lvert V_C\rvert}\}$ (otherwise, for any vertex not in a prefix, a preceding vertex can always be found to replace it). Moreover, the $k$-MDB derived from $(S,C')$ have more edges than any $k$-MDB derived from $(S, C)$. Consequently, Lemma \ref{lem:ieub} provide a correct edge upper bound.
	\end{proofsketch}

	\begin{algorithm}[!t]

		\caption{\protect \ub {$(S,C)$, $k$}}
		\scriptsize
		
		\label{alg:ub}
		\KwIn{Instance $(S=(U_S,V_S),C=(U_C,V_C))$ and integer $k$}
		\KwOut{two vertex upper bounds and an edge upper bound}
		
		$k'\leftarrow k-\lvert\overline E_S\rvert$\;
		Let $u_1,u_2,\cdots,u_{\lvert U_C\rvert}$ and $v_1,v_2,\cdots,v_{\lvert V_C\rvert}$ be the vertices of $U_C$ and $V_C$ in ascending order of $\overline d_S(\cdot)$\;
		$\overline d_U\leftarrow 0$; $\overline d_V\leftarrow 0$; $c_U\leftarrow 0$; $c_V\leftarrow 0$\;
		\For{$i\leftarrow 1\dots\lvert U_C\rvert$ s.t. $\overline d_U+\overline d_S(u_i)\le k'$}{
			$\overline d_U\leftarrow \overline d_U+\overline d_S(u_i)$\;
			$c_U\leftarrow c_U+1$\;
		}
		$e\leftarrow (\lvert U_S\rvert+c_U)\cdot \lvert V_S\rvert - \lvert\overline E_S\rvert - \overline d_U$\; $i\leftarrow c_U$\; 
		\For{$j\leftarrow 1\dots\lvert V_C\rvert$ s.t. $\overline d_V+\overline d_S(v_i)\le k'$} {
			\While{$\overline d_U+\overline d_V+j>k'$}{
				$\overline d_U\leftarrow \overline d_U-\overline d_S(u_i)$\;
				$i\leftarrow i-1$\;
			}
			$\overline d_V\leftarrow \overline d_V+\overline d_S(u_i)$\;
			$c_V\leftarrow c_V+1$\;
			$e\leftarrow\max\{e, (\lvert U_S\rvert+i)\cdot(\lvert U_S\rvert+j)-\lvert E_S\rvert-\overline d_U-\overline d_V\}$\;
		}
		
		\Return{$\lvert U_S\rvert+c_U$, $\lvert V_S\rvert+c_V$, $e$}\;

	\end{algorithm}
	
	
	\begin{detail}
		Algorithm \ref{alg:ub} presents a pseudocode of our upper-bounding techniques. It starts by setting $k'$ as the remaining number of missing edges (line 1), and sorting the vertices of $U_C$ and $V_C$ in ascending order of $\overline d_S(\cdot)$ (line 2), which can be efficiently done using a bucket sort in linear time and space. Then, the algorithm calculates the vertex upper bound in $U$ by finding the longest prefix of $U_C$ that can be added to $U_S$ (lines 4-6), where $c_U$ records the prefix length. After that, the algorithm derives the upper bound $c_V$ for vertices in $V$ using the same way, and simultaneously calculates the upper bound $e$ for edges by removing vertices from the end of the prefix of $U_C$ while computing the longest prefix of $V_C$ (lines 9-15). Finally, it returns two upper bounds of both sides as well as an edge upper bound. We present the time complexity of the algorithm in the following lemma.
	\end{detail}
	
	\begin{lemma}
		The time complexity of Algorithm \ref{alg:ub} is $O(n)$.
	\end{lemma}
	
	
	\subsection{A Heuristic Approach} \label{sec:heu}

	We propose a heuristic approach to efficiently calculate a large $k$-defective biclique. We observe that in many real-world graphs, one side of the $k$-MDB with size threshold $\theta$ contains exactly $\theta$ vertices, while the other side contains much more than $\theta$ vertices. Motivated by this, our heuristic approach 
	aims to find a large $k$-defective biclique with one side has $\theta$ vertices.
	Specifically, given a graph $G=(U,V,E)$ and an integer $k$, we initially set two vertex subset $U'=\empty$ and $V'=V$. let $u\in U$ be the vertex with largest $d_{V'}(u)$. If $G(U'\cup\{u\},V')$ has more than $k$ non-edges, we iteratively delete vertices $v\in V$ with largest $\overline d_{U'\cup\{u\}}(v)$ until $G(U'\cup\{u\},V')$ has no more than $k$ non-edges, and then we add $u$ to $U'$. We repeat this process until $\lvert U\rvert=\theta$, and yield a large $k$-MDB $G(U',V')$ if $\lvert V\rvert=\theta$. Otherwise, we yield an empty answer. It is easy to verify that this approach can be finished in $O(\theta n+m)$, which is linear as $\theta$ is always a small value.

	\subsection{Parallelization}
	
		To further enhance the scalability of our algorithms for handling larger-scale graphs, we implement parallelization strategy based on our proposed algorithms and optimization techniques. 
		The parallelization of our branch-and-bound and pivoting-based algorithms is based on the search tree of the algorithm (as shown in Fig. \ref{fig:eg2}b and Fig. \ref{fig:eg3}b). Specifically, for a given branching strategy, once an instance $(S, C)$ corresponding to a branch is determined, the sub-branches generated by this branch are uniquely defined and mutually independent. Therefore, we parallelize the processing of sub-branches at lower levels of the search tree. Additionally, the parallelization process is controlled by a global counter that tracks the total number of parallel tasks, and when the counter reaches a predefined threshold, the sub-branches generated by each branch are processed sequentially.
		Additionally, we also perform parallelization on the ordering-based reduction through the traversal of ordered vertices, and on the common-neighbor-based reduction through the loop in lines 2-13 of Algorithm \ref{alg:cnred}. 
		
		\begin{algorithm}[!t]
			\caption{\protect \mdbp}
			\scriptsize
			\label{alg:mdbpplus}
			\KwIn{A bipartite graph $G=(U,V,E)$, integers $k$, $\theta$}
			\KwOut{A $k$-MDB of $G$}
			$D^*\leftarrow$ \heu{$G$, $k$, $\theta$}\;
			$\theta_U\leftarrow \max_{u\in U}d(u)$\;
			\While{$\theta_U>\theta$}{
				$\theta_V\leftarrow\max\{\theta,\, \lfloor\lvert D^*\rvert/\theta_U\rfloor\}$;\ 
				$\theta_U\leftarrow\max\{\theta, \lfloor\theta_U/2\rfloor\}$\;
				$G'\leftarrow$ \textbf{parallel} \cnred{$G'$, $k$, $\min\{\theta_U$, $\theta_V\}$}\;
				$\{u_1,u_2,\dots,u_{\lvert U\rvert}\}\leftarrow$ the descending degree order of $U$\;
				\textbf{parallel} \For{$i\in 1\dots\lvert U\rvert$} { 
					$S\leftarrow (\{u_i\}, \emptyset)$;\ 
					$C\leftarrow (N_{+}^2(u_i),\ \bigcup_{v\in N_{+}^2(u_i)}N(v))$\;
					$C\leftarrow$ \redi{$C$, $u_i$, $\theta_U$, $\theta_V$}\;
					\brpplus{$S$, $C$}; \tcp{\brp equipped with graph reduction techniques, upper-bounding techniques and parallelization }
				}
			}
			\Return{$D^*$}\;

			\Fn{\redi{$C=(U_C,V_C)$, $u$, $\theta_U$, $\theta_V$}} {
				$U_{C'}\leftarrow \left\{v\in U_C\mid cn(u,v)\ge\theta_V-k\right\}$\;
				$V_{C'}\leftarrow \left\{v\in V_C\mid d(v)\ge\theta_U-k\right\}$\;
				\Return $(U_{C'}, V_{C'})$
			}

		\end{algorithm}
	
	\subsection{Optimized Algorithms} \label{sec:oa}

	All the proposed optimization techniques can be applied to our branch-and-bound algorithm and pivoting-based algorithm. We refer to the branch-and-bound algorithm and pivoting-based algorithm equipped with all the optimization techniques as \mdbb and \mdbp, respectively. An implementation of \mdbp is presented in Algorithm \ref{alg:mdbpplus}. The algorithm begins by obtaining an initial $k$-defective biclique (line 1) by performing the heuristic approach. Then it starts to find $k$-MDB by performing the progressive bounding reduction (lines 2-4). In each iteration, two size thresholds $\theta_U$ and $\theta_V$ are generated based on Eq. (\ref{eq:6}) and Eq. (\ref{eq:7}). Then, the parallelized common-neighbor-based reduction are employed on the input graph $G$ (line 5) by invoking \cnred (Algorithm \ref{alg:cnred}). the algorithm continues by performing the parallelized ordering-based reduction on the reduced graph (lines 6-10). In detail, it traverses each vertex $u$ of $U$ in descending degree order (lines 7-8), and constructs an initial instance $(S, C)$. After that, the algorithm further reduce the size of $C$ by invoking \redi (line 9), which utilizes the common-neighbor-based reduction (lines 13-14). Finally, the algorithm runs \brpplus to find $k$-MDB in $(S, C)$, which is adapted from \brp in Algorithm \ref{alg:mdbp} by additionally implementing the graph reduction techniques and the upper-bounding techniques (Algorithm \ref{alg:ub}). By taking the maximum returned solution from all invocations of \brpplus, the algorithm derives a $k$-MDB $D^*$ at the end and returns $D^*$ as the final answer (line 11).
	The implementation of \mdbb is quite similar to that of \mdbp, which differs only in that replacing \brpplus in Algorithm \ref{alg:mdbpplus} with \brb from Algorithm \ref{alg:mdbb} incorporating the optimization techniques used in \brpplus.

	
	\section{Experiments}
	
	\subsection{Experimental Setup}
	\textbf{Datasets.}
	We collect 11 real-world bipartite graphs from diverse categories to evaluate the efficiency of our algorithms. These graphs have been widely used as benchmark datasets for evaluating the performance of other bipartite subgraph search problem \cite{RN13,RN71,RN66}. 
	The detailed information of these graphs are shown in Table \ref{tab:dat}, where $\Delta$ and $\overline d$ represent the maximum degree and average degree of vertices on each side, respectively. All the datasets are available at \url{http://konect.cc/networks}.

	\noindent\textbf{Algorithms.}
	We implement two proposed algorithm \mdbb and \mdbp, which are respectively the branch-and-bound algorithm in Algorithm \ref{alg:mdbb} and the pivoting-based algorithm in Algorithm \ref{alg:mdbp} equipped with all our proposed optimization techniques in Sec. \ref{sec:ot}.
	We also implement two baseline algorithms, denoted as \mdbase and \mdcadp for comparative analysis. Specifically, \mdbase is a variant of \mdbb with a time complexity of $\mathcal O^*(2^n)$, which is implemented by substituting the binary branching strategy at line 9-10 in Algorithm \ref{alg:mdbb} with the straightforward strategy that randomly selecting a branching vertex $u$ from $C$. On the other hand, \mdcadp is adapted from the state-of-the-art maximal defective clique algorithm \pivotplus \cite{RN64}. In detail, \mdcadp converts the input bipartite graph into a traditional graph by treating all vertices on the same side as fully connected and runs \pivotplus on the converted graph, subsequently extracting a $k$-MDB from all the maximal defective cliques of the converted graph. Furthermore, both \mdbase and \mdcadp incorporate the optimization techniques outlined in Sec. \ref{sec:red} and \ref{sec:heu}. All algorithms are implemented in C++ and executed on a PC equipped with a 2.2GHz CPU and 128GB RAM running Linux.
	
	\noindent\textbf{Parameter settings.}
	During the evaluations of all algorithms, we vary the value of $k$ from $1$ to $6$. As outlined in Sec. \ref{sec:pre}, we focus on finding a connected $k$-MDB, achieved by setting $\theta> k$. Based on this, we vary the value of $\theta$ from $k+1$ to $10$ for a specified $k$. 
	
	
	\begin{table}[!t]
		\smaller[1.5]
		\caption{Datasets used in experiments}
		\renewcommand{\arraystretch}{.95}
		\setlength{\tabcolsep}{4pt}
		\label{tab:dat}
		\begin{tabular}{lccccc}
			\toprule
			Dataset&$\lvert U \rvert$&$\lvert V\rvert$&$\lvert E\rvert$&{$\Delta$}&mdbk{$\overline d$}\\
			\midrule
			LKML		&42,045		&337,509	&599,858	&(31,719, 6,627)	&(14, 1)\\
			Team		&901,130	&34,461		&1,366,466	&(17, 2,671)		&(1, 39)\\
			Mummun		&175,214	&530,418	&1,890,661	&(968, 19805)		&(10, 3)\\
			Citeu		&153,277	&731,769	&2,338,554	&(189,292, 1,264)	&(15, 3)\\
			IMDB		&303,617	&896,302	&3,782,463	&(1,334, 1,590)		&(12, 4)\\
			Enwiki		&18,038		&2,190,091	&4,129,231	&(324,254, 1,127)	&(228, 1)\\
			Amazon		&2,146,057	&1,230,915	&5,743,258	&(12,180, 3,096)	&(2, 4)\\
			Twitter		&244,537	&9,129,669	&10,214,177	&(640, 18,874)		&(41, 1)\\
			Aol			&4,811,647	&1,632,788	&10,741,953	&(100,629, 84,530)	&(2, 6)\\
			Google		&5,998,790	&4,443,631	&20,592,962	&(423, 95,165)		&(3, 4)\\
			LiveJournal	&3,201,203	&7,489,073	&112,307,385	&(300, 1,053,676)	&(35, 14)\\

			\bottomrule
			\vspace{-10pt}
		\end{tabular}
	\end{table}

	\begin{table*}[!t]
		\caption{Results of runtime among different algorithms on 10 real-world graphs (in seconds)}
		\smaller[3]
		\centering
		\label{tab:rt}
		\renewcommand{\arraystretch}{0.95}
		\begin{tabular}{c|c|c|c|cccc|c|c|c|cccc|c|c|c|cccc}
			\hline
			
			Dataset&$k$&$\theta$&$\lvert E^*\rvert$&\mdbp&\mdbb&\mdbase&\mdcadp&$k$&$\theta$&$\lvert E^*\rvert$&\mdbp&\mdbb&\mdbase&\mdcadp&$k$&$\theta$&$\lvert E^*\rvert$&\mdbp&\mdbb&\mdbase&\mdcadp\\
			\hline
			\multirow{4}{*}{LKML} 
			& \multirow{2}{*}{1} & 2 & 1,781 & \textbf{13.03} & 17.08 & 111.82 &- & \multirow{2}{*}{2} & 3 & 460 & \textbf{43.72} & - & - & - & \multirow{2}{*}{3} & 4 & 149 & \textbf{267.01} & - & - & - \\ 
			&& 5 & 84 & \textbf{1.61} & 2.15 & 2.36 & 28.19 && 6 & 70 & \textbf{7.60} & 30.14 & 357.33 & - && 7 & 60 & \textbf{9.53} & 37.96 & 1,082.91 & - \\
			\cline{2-22}
			& \multirow{2}{*}{4} & 7 & 66 & \textbf{103.03} & 1,216.70 & - & - & \multirow{2}{*}{5} & 8 & 67 & \textbf{107.31} & 1,224.86 & - & - & \multirow{2}{*}{6} & 9 & 0 & \textbf{248.95} & 1,732.26 & - & - \\ 
			&& 8 & 60 & \textbf{9.36} & 46.38 & 2,360.74 & - && 9 & 0 & \textbf{13.42} & 51.55 & 3,441.79 & - && 10 & 0 & \textbf{12.29} & 46.38 & 4,162.94 & - \\ 
			\hline 
			\multirow{4}{*}{Team} 
			& \multirow{2}{*}{1} & 2 & 701 & \textbf{19.26} & 39.98 & 43.67 & - & \multirow{2}{*}{2} & 3 & 265 & \textbf{45.51} & - & - & - & \multirow{2}{*}{3} & 4 & 149 & \textbf{154.56} & - & - & - \\ 
			&& 5 & 99 & \textbf{7.14} & 7.19 & 7.73 & 66.18 && 6 & 64 & \textbf{9.53} & 14.09 & 31.43 & - && 7 & 0 & \textbf{10.07} & 15.41 & 242.20 & - \\
			\cline{2-22}
			& \multirow{2}{*}{4} & 5 & 111 & \textbf{591.71} & - & - & - & \multirow{2}{*}{5} & 6 & 79 & \textbf{2,684.57} & - & - & - & \multirow{2}{*}{6} & 8 & 0 & \textbf{927.43} & - & - & - \\ 
			&& 8 & 0 & \textbf{9.92} & 15.89 & 1,734.29 & - && 9 & 0 & \textbf{10.04} & 16.35 & 9,791.32 & - && 10 & 0 & \textbf{7.65} & \textbf{11.21} & - & - \\ 
			\hline 
			\multirow{4}{*}{Mummun}
			& \multirow{2}{*}{1} & 2 & 10,315 & \textbf{18.49} & 18.62 & 79.28 & - & \multirow{2}{*}{2} & 3 & 8,803 & \textbf{22.88} & 8,661.39 & - & - & \multirow{2}{*}{3} & 4 & 4,185 & \textbf{39.95} & - & - & - \\ 
			&& 5 & 2,469 & \textbf{13.57} & 13.59 & 16.49 & - && 6 & 1,768 & \textbf{21.08} & 37.83 & 84.27 & - && 7 & 1,145 & \textbf{33.58} & 916.16 & 3,677.02 & -  \\
			\cline{2-22}
			& \multirow{2}{*}{4} & 5 & 2,481 & \textbf{45.73} & - & - & - & \multirow{2}{*}{5} & 6 & 1,783 & \textbf{69.39} & - & - & - & \multirow{2}{*}{6} & 7 & 1,163 & \textbf{160.89} & - & - & - \\ 
			&& 8 & 878 & \textbf{38.18} & 5,538.54 & 10,440.17 & - && 9 & 886 & \textbf{40.64} & 6,737.14 & - & - && 10 & 804 & \textbf{53.09} & - & - & -  \\
			\hline 
			\multirow{4}{*}{Citeu}
			& \multirow{2}{*}{1} & 2 & 12,425 & \textbf{141.21} & 146.81 & -& - & \multirow{2}{*}{2} & 3 & 12,427 & \textbf{148.38} & 187.69 & - & - & \multirow{2}{*}{3} & 4 & 9,989 & \textbf{137.15} & 4,332.55 & - & - \\ 
			&& 5 & 7,394 & \textbf{5.48} & 5.69 & - & - && 6 & 4,960 & \textbf{7.88} & 77.20 & - & - && 7 & 3,063 & \textbf{17.92} & 3,500.77 & - & -  \\
			\cline{2-22}
			& \multirow{2}{*}{4} & 5 & 7,406 & \textbf{141.78} & - & - & - & \multirow{2}{*}{5} & 6 & 4,975 & \textbf{158.43} & - & - & - & \multirow{2}{*}{6} & 7 & 3,081 & \textbf{253.15} & - & - & - \\ 
			&& 8 & 2,600 & \textbf{25.92} & - & - & - && 9 & 2,600 & \textbf{36.03} & - & - & - && 10 & 2,614 & \textbf{134.21} & - & - & -  \\
			\hline 
			\multirow{4}{*}{IMDB}
			& \multirow{2}{*}{1} & 2 & 775 & 5.99 & \textbf{5.39} & 4,460.09 & - & \multirow{2}{*}{2} & 3 & 562 & \textbf{8.86} & 16.96 & - & - & \multirow{2}{*}{3} & 4 & 522 & \textbf{15.85} & 436.35 & - & - \\ 
			&& 5 & 514 & 6.12 & \textbf{5.35} & 145.10 & 23.43 && 6 & 514 & \textbf{6.76} & 6.91 & - & 1,123.22 && 7 & 482 & \textbf{14.73} & 25.23 & - & -  \\
			\cline{2-22}
			& \multirow{2}{*}{4} & 5 & 526 & \textbf{19.99} & 593.86 & - & - & \multirow{2}{*}{5} & 6 & 529 & \textbf{31.05} & 683.75 & - & - & \multirow{2}{*}{6} & 7 & 498 & \textbf{57.45} & 7,936.24 & - & - \\ 
			&& 8 & 436 & \textbf{15.13} & 67.00 & - & - && 9 & 436 & \textbf{19.34} & 170.90 & - & - && 10 & 426 & \textbf{22.77} & 402.50 & - & -  \\
			\hline
			\multirow{4}{*}{Enwiki}
			& \multirow{2}{*}{1} & 2 & 203,569 & \textbf{3,630.07} & - & - & - & \multirow{2}{*}{2} & 3 & 133,759 & \textbf{4,977.99} & - & - & - & \multirow{2}{*}{3} & 4 & 80,461 & \textbf{3,259.88} & - & - & - \\ 
			&& 5 & 21,939 & \textbf{156.70} & 158.87 & 354.17 & 387.15 && 6 & 4,240 & \textbf{199.76} & 582.63 & - & - && 7 & 3,224 & \textbf{264.30} & - & - & -  \\
			\cline{2-22}
			& \multirow{2}{*}{4} & 5 & 21,951 & \textbf{3,826.16} & - & - & - & \multirow{2}{*}{5} & 6 & 4,255 & \textbf{5,744.46} & - & - & - & \multirow{2}{*}{6} & 8 & 2,330 & \textbf{6,557.63} & - & - & - \\ 
			&& 8 & 2,316 & \textbf{339.60} & - & - & - && 9 & 1,696 & \textbf{394.81} & - & - & - && 10 & 1,274 & \textbf{639.63} & - & - & -  \\
			\hline
			\multirow{4}{*}{Amazon}
			& \multirow{2}{*}{1} & 2 & 1,319 & \textbf{42.68} & 45.42 & - & - & \multirow{2}{*}{2} & 3 & 427 & \textbf{56.54} & 646.44 & - & - & \multirow{2}{*}{3} & 4 & 423 & \textbf{88.75} & 1,195.66 & - & - \\ 
			&& 5 & 413 & 11.71 & \textbf{11.12} & 12.21 & 49.70 && 6 & 418 & 12.91 & \textbf{12.00} & 167.48 & - && 7 & 357 & \textbf{13.15} & 14.91 & 183.57 & - \\
			\cline{2-22}
			& \multirow{2}{*}{4} & 5 & 428 & \textbf{99.80} & 850.49 & - & - & \multirow{2}{*}{5} & 6 & 433 & \textbf{249.97} & - & - & - & \multirow{2}{*}{6} & 7 & 390 & \textbf{268.21} & - & - & - \\ 
			&& 8 & 368 & \textbf{13.68} & 15.96 & 171.44 & - && 9 & 379 & \textbf{16.09} & 20.38 & 319.93 & - && 10 & 390 & \textbf{18.63} & 25.61 & - & -  \\
			\hline
			\multirow{4}{*}{Twitter}
			& \multirow{2}{*}{1} & 2 & 5,369 & \textbf{11.13} & 13.24 & 46.13 & - & \multirow{2}{*}{2} & 3 & 5,383 & \textbf{10.88} & 51.42 & 2,143.04 & - & \multirow{2}{*}{3} & 4 & 5,397 & \textbf{15.47} & 320.65 & - & - \\ 
			&& 5 & 5,369 & \textbf{6.27} & 7.60 & 36.58 & 8,976.29 && 6 & 5,383 & \textbf{6.63} & 19.51 & 1,432.31 & - && 7 & 5,397 & \textbf{9.76} & 84.34 & - & -  \\
			\cline{2-22}
			& \multirow{2}{*}{4} & 5 & 5,411 & \textbf{20.75} & 3,549.73 & - & - & \multirow{2}{*}{5} & 6 & 5,425 & \textbf{31.20} & - & - & - & \multirow{2}{*}{6} & 7 & 5,439 & \textbf{57.67} & - & - & - \\ 
			&& 8 & 5,411 & \textbf{13.29} & 491.74 & - & - && 9 & 5,425 & \textbf{23.69} & 3,268.33 & - & - && 10 & 5,439 & \textbf{43.93} & 9,894.96 & - & -  \\
			\hline
			\multirow{4}{*}{Aol}
			& \multirow{2}{*}{1} & 2 & 10,869 & 1,032.95 & \textbf{970.79} & 10,736.08 & - & \multirow{2}{*}{2} & 3 & 3,208 & \textbf{1,635.79} & - & - & - & \multirow{2}{*}{3} & 4 & 853 & \textbf{3,122.56} & - & - & - \\ 
			&& 5 & 334 & 49.77 & \textbf{48.69} & 49.84 & 5,703.42 && 6 & 238 & \textbf{94.15} & 181.85 & 581.87 & - && 7 & 207 & \textbf{123.04} & 408.97 & 8,207.80 & - \\
			\cline{2-22}
			& \multirow{2}{*}{4} & 5 & 346 & \textbf{9,271.83} & - & - & - & \multirow{2}{*}{5} & 7 & 219 & \textbf{2,388.21} & - & - & - & \multirow{2}{*}{6} & 9 & 174 & - & - & - & - \\ 
			&& 8 & 188 & \textbf{470.19} & 1,011.16 & - & - && 9 & 166 & \textbf{4,496.57} & 5,115.84 & - & - && 10 & 164 & \textbf{5,585.25} & 8,636.22 & - & -  \\
			\hline
			\multirow{4}{*}{Google}
			& \multirow{2}{*}{1} & 2 & 10,829 & \textbf{2,163.49} & 2,525.77 & 4,814.91 & - & \multirow{2}{*}{2} & 3 & 3,370 & \textbf{4,121.31} & - & - & - & \multirow{2}{*}{3} & 4 & 1,201 & \textbf{6,250.09} & - & - & - \\ 
			&& 5 & 439 & \textbf{552.42} & 573.02 & 629.33 & - && 6 & 184 & \textbf{898.04} & 2,771.57 & 4,077.91 & - && 7 & 151 & \textbf{1,351.14} & - & - & -  \\
			\cline{2-22}
			& \multirow{2}{*}{4} & 7 & 157 & \textbf{5,331.22} & - & - & - & \multirow{2}{*}{5} & 8 & 155 & \textbf{6,750.40} & - & - & - & \multirow{2}{*}{6} & 9 & 156 & \textbf{8,756.03} & - & - & - \\ 
			&& 9 & 140 & \textbf {613.17} & 9,387.13 & - & - && 10 & 145 & \textbf{683.84} & - & - & - && 10 & 145 & \textbf{1,978.33} & - & - & -  \\
			\hline

			\hline
			
		\end{tabular}
		\vspace{-10pt}
	\end{table*}

	
	\subsection{Performance Studies}

	\noindent\textbf{Exp1: Efficiency comparison among different algorithms.}
	We evaluate the efficiency of two proposed algorithms \mdbb and \mdbp, along with two baseline algorithms \mdbase and \mdcadp. The results are shown in Table \ref{tab:rt}, where the column $\lvert E^*\rvert$ represents the number of edges in the $k$-MDB, and the notion ``-'' represents a timeout exceeding 3 hours.	
	As observed, \mdbp shows superior performance across most input parameters, achieving two order of magnitude faster than \mdcadp on most datasets, and two order of magnitude faster than \mdbase on most datasets when $k>3$. Notably, \mdbp essentially solves all tests on most of datasets when setting $k\ge 3$, while \mdbase and \mdcadp can hardly finish any tests within the time limit. For example, on the \textit{Mummun} dataset with $k=3$ and $\theta=7$, \mdbase takes 3,677 seconds and \mdcadp times out at 108,000 seconds, whereas \mdbp completes in just 34 seconds, yielding speedups of $100\times$ over \mdbase and $1000\times$ over \mdcadp. On the other hand, when $k\le 2$ and $\theta\ge 5$, the efficiency of \mdbp is comparable to \mdbb. However, for $k\ge 3$ or $\theta\le 4$, \mdbp demonstrates a significant advantage. For example, on the \textit{Twitter} dataset with $k=1$ and $\theta=5$, both \mdbp and \mdbb runs for about 7 seconds, but with $k=4$ and $\theta=5$, \mdbb takes about 1 hour while \mdbp completes in just 20 seconds, achieving a speedup of more than $150\times$. These results underscore the effectiveness and high efficiency of the pivoting-based branching strategy mentioned in Sec. \ref{sec:mdbp}. Additionally, \mdbb retains advantages over \mdbase and \mdcadp, highlighting the effectiveness of the new binary branching strategy in Sec. \ref{sec:mdbb}.

	\noindent\textbf{Exp2: Runtime with varying $\theta$ and $k$.}
	We evaluate the performance of 4 algorithms \mdbb, \mdbp, \mdbase and \mdcadp with varying $\theta$ and $k$ on two representative datasets, \textit{IMDB} and \textit{Google}. The results are shown in Fig. \ref{fig:var}, where the notion ``INF'' represents the runtime exceeds 3 hours. As observed, both \mdbp and \mdbb consistently outperform baseline algorithms \mdbase and \mdcadp with varying $\theta$ and $k$, with \mdbp achieving speedups greater than $1000\times$. For example, in Fig. \ref{fig:var}a, \mdbp complete in 10 seconds at $\theta=3$, while \mdbase and \mdcadp cannot finish within 3 hours. When $\theta\ge 8$ or $k\le 2$, \mdbp and \mdbb have similar runtimes, demonstrating the efficacy of our optimization techniques in reducing graph size. 
	Furthermore, as $k$ increases or $\theta$ decreases, the time curve of \mdbp shows a gentler slope compared to other algorithms, indicating its lower sensitivity to input parameters and superior scalability. 
	Additionally, we evaluate the performance of \mdbb and \mdbp in searching for the maximum biclique by setting $k=0$. We compare them with a state-of-the-art maximum biclique search algorithm \mbc \cite{RN13}. The results are shown in Fig. \ref{fig:var}e and Fig. \ref{fig:var}f. As observed, \mdbp and \mdbb exhibit an advantage over \mbc due to two main factors: (1) both of our algorithms achieves better time complexities than \mbc, which has a time complexity of $O^*(2^n)$; and (2) our algorithms incorporate more advanced heuristic approach and graph reduction techniques, further enhancing their efficiency. 
	
	\begin{figure}[!t]
		\centering
		\setlength{\belowcaptionskip}{-0.6cm} 
		\includegraphics[width=0.9\columnwidth,height=8.25cm]{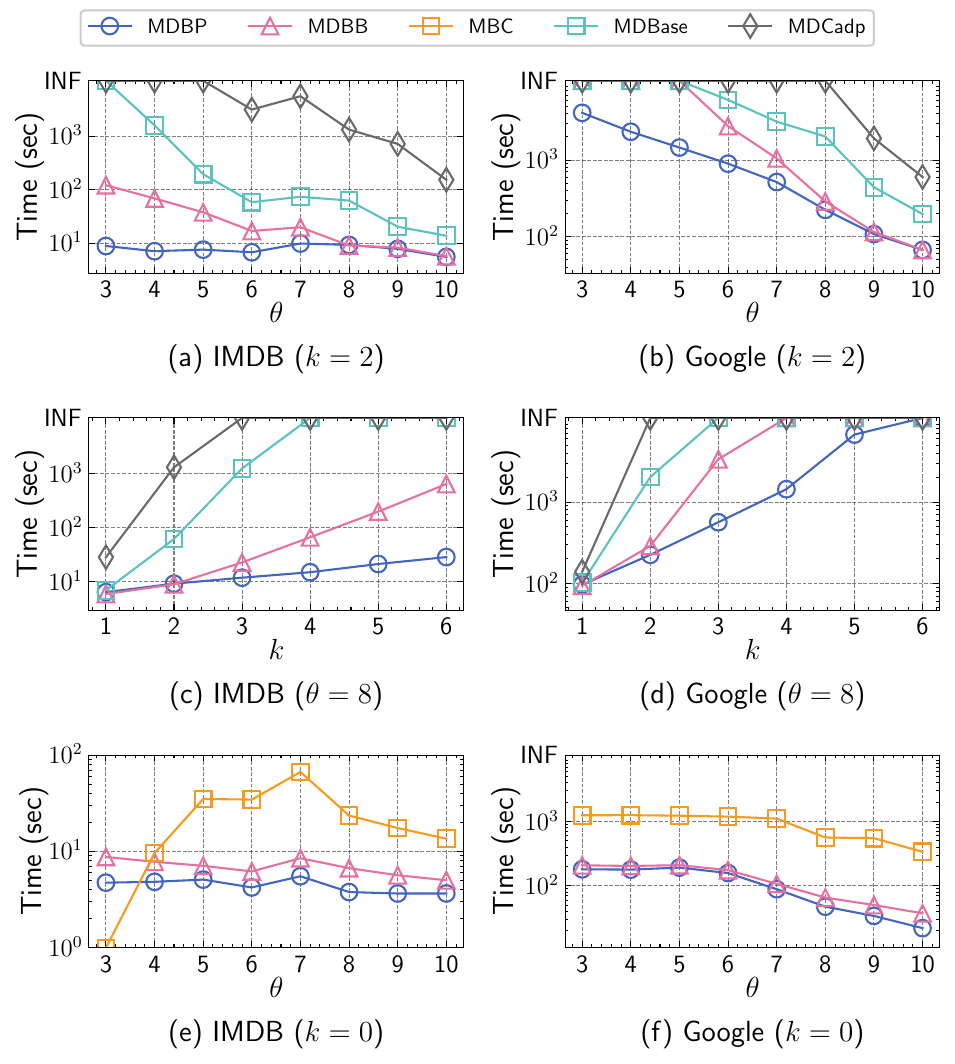}
		\caption{Runtime with varying $\theta$ and $k$}
		\label{fig:var}
	\end{figure}

	\noindent\textbf{Exp3: Efficacy of the heuristic approach.}
	To evaluate the effectiveness of our proposed heuristic approach, we compare \mdbb and \mdbp with their versions without the heuristic approach, referred to as \mdbbnh and \mdbpnh, respectively. The results are shown in Fig. \ref{fig:heu}. As observed, \mdbb and \mdbp always outperform \mdbbnh and \mdbpnh, respectively. This is because the heuristic approach provides a relatively large $k$-defective biclique at the beginning, which enhances the effectiveness of our upper-bounding techniques. Additionally, we observe that \mdbpnh still always performs better than \mdbb when $\theta\le 6$ or $k\ge 3$, which further highlights the high efficiency of our pivoting techniques.

	\begin{figure}[!t]
		\centering
		\setlength{\belowcaptionskip}{-0.4cm} 
		\includegraphics[width=0.9\columnwidth,height=5.5cm]{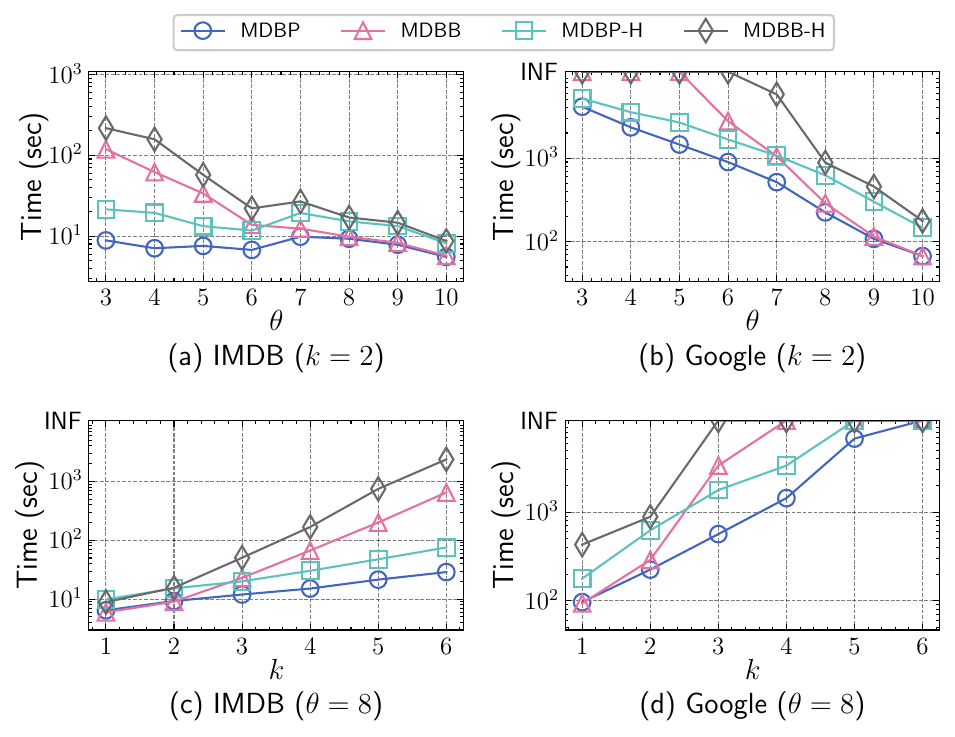}
		\caption{Runtime without using the heuristic approach}
		\label{fig:heu}
	\end{figure}
	
	\noindent\textbf{Exp4: Efficacy of graph reduction techniques.}
	To evaluate the effectiveness of our graph reduction techniques, we implement 4 ablation versions of \mdbp and \mdbb by successively removing the following reductions: one-non-neighbor reduction, common-neighbor-based reduction, progressive bounding reduction and ordering-based reduction. The results are presented in Fig \ref{fig:gr}, which label the original algorithm as ``Original'', and label the 4 ablation versions as ``-1'', ``-1C'', ``-1CP'', ``-1CPO''. As observed, all graph reduction techniques contribute to efficiency improvements. Specifically, the one-non-neighbor and common-neighbor-based reductions yield stable performance improvements with varying $\theta$ and $k$, while the progressive bounding and ordering-based reduction demonstrate more pronounced improvements when $\theta$ decreases or $k$ increases. Notably, when $k=1$ or $\theta=10$, versions of \mdbb and \mdbp without progressive bounding reduction still outperform those with it. This suggests that our algorithms are well-designed and robust, achieving strong performance for small $k$ or large $\theta$ even without progressive bounding reduction.
	
	\begin{figure}[!t]
		\centering
		\setlength{\belowcaptionskip}{-0.4cm} 
		\includegraphics[width=0.9\columnwidth,height=5.5cm]{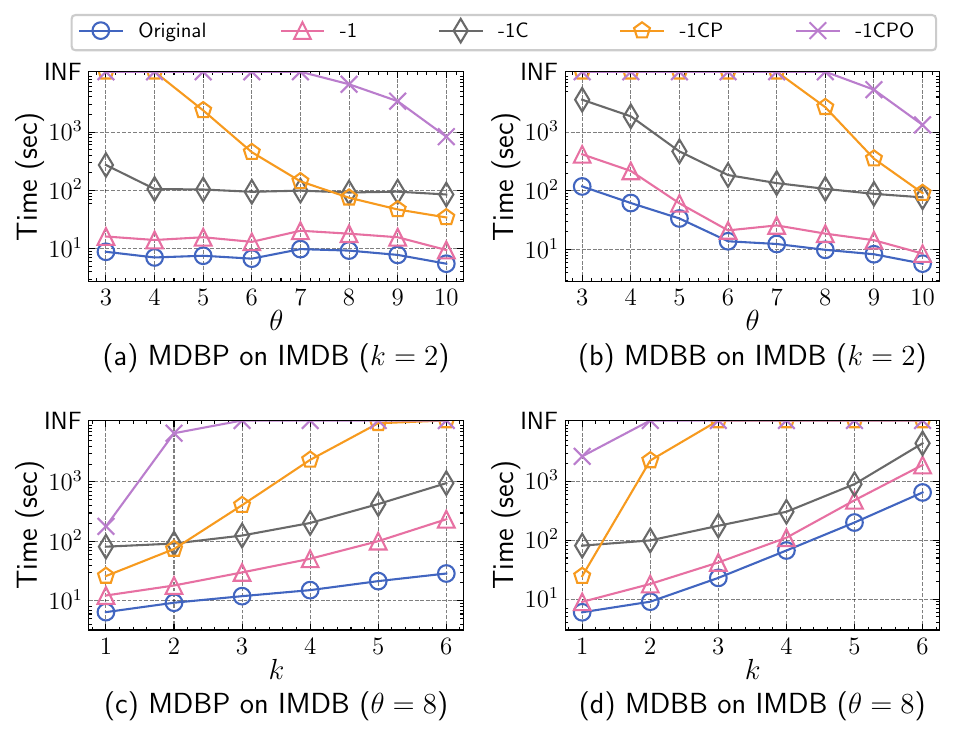}
		\caption{Runtime without graph reduction techniques}
		\label{fig:gr}
	\end{figure}

	\begin{figure}[!t]
		\centering
		\setlength{\belowcaptionskip}{-0.4cm} 
		\includegraphics[width=0.9\columnwidth,height=5.5cm]{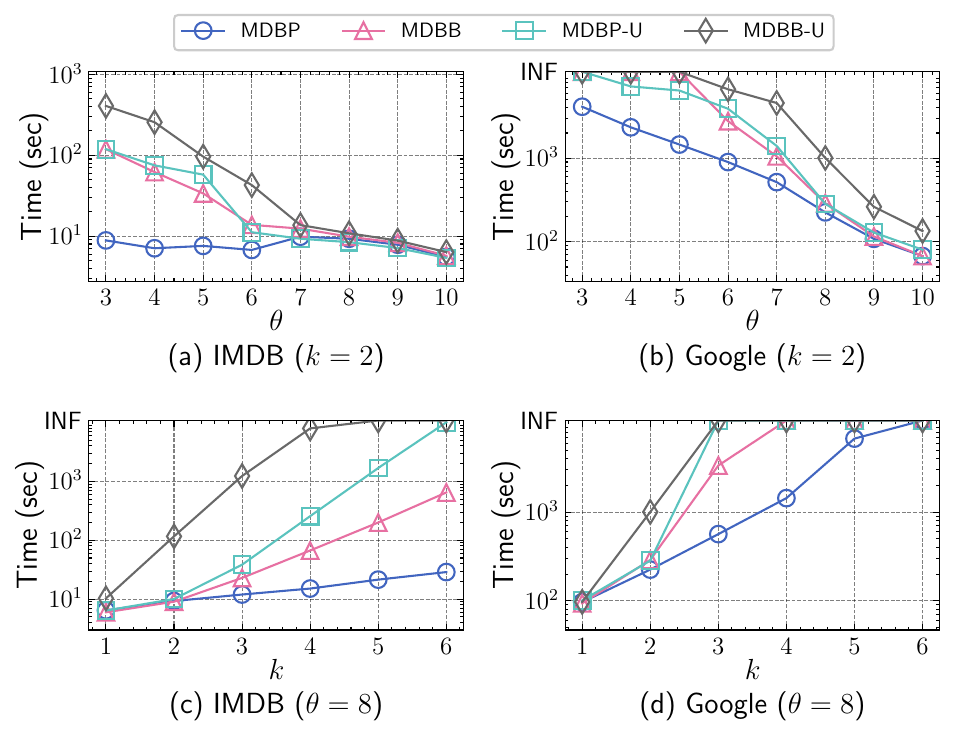}
		\caption{Runtime without upper-bounding techniques}
		\label{fig:ub}
	\end{figure}
	
	\noindent\textbf{Exp5: Efficacy of upper-bounding techniques.}
	To evaluate the effectiveness of our proposed upper-bounding techniques, we compare \mdbb and \mdbp with their versions excluding the upper-bounding techniques (Algorithm \ref{alg:ub}), denoted as \mdbbnh and \mdbpnh, respectively. The results are shown in Fig. \ref{fig:ub}. As observed, incorporating the upper-bounding techniques significantly improves the efficiency of both \mdbb and \mdbp. Moreover, the improvement becomes more pronounced as $k$ increases or $\theta$ decreases, indicating the strong effectiveness and efficiency of our proposed upper-bounding techniques. Additionally, the improvement of \mdbp by performing upper-bounding techniques is more substantial compared to \mdbb. This is likely because the pivoting-based branching strategy introduces more non-edges to the partial set $S$ earlier, resulting in tighter vertex and edge upper bounds derived from the upper-bounding techniques, which helps prune more unnecessary instances.
	
	\begin{table}
		\caption{\centering{Runtime on synthetic graphs (in seconds)}}
		\smaller[2.5]
		\centering
		\label{tab:syn}
		\setlength{\tabcolsep}{0.8pt}
		\begin{tabular}{c|c|c|c|ccc|ccc|ccc}
				\hline
				\multirow{2}{*}{Distribution}&\multirow{2}{*}{$\rho$}&\multirow{2}{*}{$\Delta$}&\multirow{2}{*}{$\overline d$}&\multicolumn{3}{c|}{$k=1,\theta=2$} &\multicolumn{3}{c|}{$k=3,\theta=4$}&\multicolumn{3}{c}{$k=5,\theta=6$}\\
				\cline{5-13}
				&&&&$\lvert E^*\rvert$&\mdbp&\mdbb&$\lvert E^*\rvert$&\mdbp&\mdbb&$\lvert E^*\rvert$&\mdbp&\mdbb\\
				\hline
				\multirow{4}{*}{Power-law}&0.2&(97, 82)&(19, 19)&197&\textbf{0.10}&0.18&181&\textbf{0.29}&0.96&175&\textbf{8.03}&11.82\\
				
				&0.3&(77, 99)&(29, 29)&269&\textbf{0.11}&0.33&277&\textbf{0.56}&9.52&292&\textbf{16.02}&133.29\\
				
				&0.4&(95, 100)&(36, 36)&349&\textbf{0.27}&\textbf{0.51}&357&\textbf{1.76}&26.01&375&\textbf{33.84}&432.58\\
				&0.5&(100, 99)&(45, 45)&482&\textbf{0.92}&\textbf{1.94}&494&\textbf{5.83}&105.49&506&\textbf{189.92}&2,030.16\\
				\hline
				\multirow{4}{*}{Normal}&0.2&(29, 29)&(17, 17)&27&\textbf{0.26}&0.47&29&\textbf{31.74}&75.84&32&\textbf{59.74}&153.87\\
				
				&0.3&(43, 43)&(29, 29)&44&\textbf{1.24}&3.04&45&\textbf{179.49}&397.23&49&\textbf{312.55}&1,225.26\\
				
				&0.4&(48, 51)&(39, 39)&56&\textbf{3.78}&10.73&67&\textbf{489.87}&1,015.43&65&\textbf{1,300.01}&-\\
				&0.5&(62, 61)&(49, 49)&65&\textbf{15.03}&43.81&65&\textbf{2,095.58}&-&-&-&-\\
				\hline
				
		\end{tabular}
		\vspace{-10pt}
	\end{table}
	
	\begin{figure}[!t]
		\centering
		\setlength{\belowcaptionskip}{-0.4cm} 
		\includegraphics[width=0.9\columnwidth, height=5.5cm]{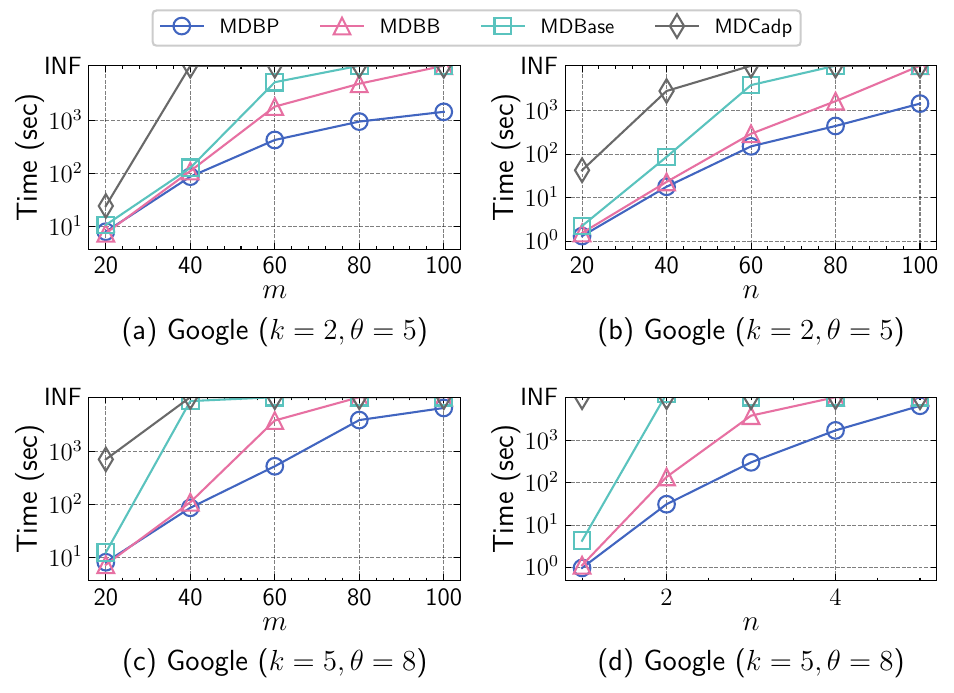}
		\caption{Scalability of different algorithms}
		\label{fig:sc}
	\end{figure}
	
	\noindent\textbf{Exp6: Performance on synthetic graphs.} 
	To discover the relationship between the efficiency of our approaches and graph characteristics, we evaluate the runtime of \mdbb and \mdbp on synthetic graphs with varying vertex degree distributions and graph densities. Specifically, we generate 8 synthetic graphs with 100 vertices on each side. Their vertex degrees follows a power-law or normal distribution, and the density $\rho=\frac{\lvert E\rvert}{\lvert U\rvert\cdot\lvert V\rvert}$ varies from 0.2 to 0.5. The results are presented in Table \ref{tab:syn}. As observed, our algorithms consistently exhibits high efficiency on power-law distribution graphs, whereas their performance on normal distribution graphs is more significantly influenced by $\rho$ and $k$. This phenomenon can be attributed to two key factors. First, Power-law graphs exhibit a larger gap between $\Delta$ and $\overline d$ than normal graphs, which facilitates our branching strategies to eliminate more redundant branches and enables our graph reduction techniques to exclude more vertices. This explains why our approaches are highly efficient on large-scale real-world graphs with a significantly gap between $\Delta$ and $\overline d$, as illustrated in Table \ref{tab:dat} and \ref{tab:rt}. Second, the higher number of $k$-MDB edges in power-law graphs benefits our upper bound techniques for effectively pruning suboptimal branches, which can also reflected in the real-world graphs in Table \ref{tab:rt}. Additionally, We observe that \mdbp outperforms \mdbb on these synthetic graphs, indicating the effectiveness advancement of our pivoting-based branching strategy.

	\noindent\textbf{Exp7: Scalability testing.}
	To evaluate the scalability of algorithms, we test the runtime of \mdbb, \mdbp, \mdbase and \mdcadp on the  \textit{Google} dataset using random sampling of vertices or edges. Results on the other datasets are consistent. We prepare 8 proportionally sampled subgraphs of \textit{Google} by sampling 20\%, 40\%, 60\% and 80\% vertices or edges from the original graph. The results are shown in Fig. \ref{fig:sc}. As observed, \mdbp shows the slowest runtime growth as the graph size increases, while \mdcadp exhibits the steepest rise; \mdbb and \mdbase fall in between, where \mdbb has better performances. This result indicates that \mdbp achieves the best scalability among the algorithms. Furthermore, regardless of graph size, both \mdbp and \mdbb significantly outperform \mdbase and \mdcadp, underscoring the effectiveness of the proposed branching strategies and optimization techniques.

	\begin{table}
		\caption{\centering{Runtime of parallelized algorithms (in seconds)}}
		\smaller[1.5]
		\setlength{\tabcolsep}{3pt}
		\centering
		\label{tab:par}
		\begin{tabular}{c|c|c|ccc|ccc}
				\hline
				\multirow{2}{*}{Dataset}&\multirow{2}{*}{$k$}&\multirow{2}{*}{$\theta$}&\multicolumn{3}{c|}{\mdbp with threads} &\multicolumn{3}{c}{\mdbb with threads}\\
				\cline{4-9}
				&&&1&10&20&1&10&20\\
				\hline
				\multirow{3}{*}{Google}&1&2&2,163.49&425.76&34.69&3,052.03&870.92&305.05\\
				
				&3&4&6,250.09&1,324.98&802.67&-&-&10,649.47\\
				
				&5&6&-&-&8,947.28&-&-&-\\
				\hline
				\multirow{3}{*}{LiveJournal}&1&2&119.61&89.14&47.11&-&-&9,314.33\\
				
				&3&4&3,996.77&997.90&850.52&-&-&-\\
				
				&5&6&-&-&10,016.72&-&-&-\\
				\hline
				
		\end{tabular}
		\vspace{-5pt}
	\end{table}
	
	\begin{figure}[!t]
		\centering
		\setlength{\belowcaptionskip}{-0.4cm} 
		\includegraphics[width=0.9\columnwidth, height=2.5cm]{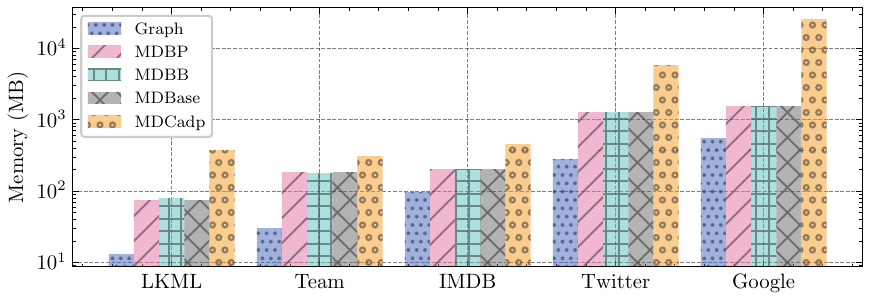}
		\caption{Memory usage of various algorithms}
		\label{fig:mem}
	\end{figure}

	\noindent\textbf{Exp8: Efficiency of parallelized algorithms.} 
		To evaluate the efficiency of the parallel version of \mdbb and \mdbp, we test their runtime on two large-scale graphs, \textit{Google} and \textit{LiveJournal}, which are two hard cases for our sequential algorithms as shown in Table \ref{tab:dat}. The results are shown in Table \ref{tab:par}. As observed, the runtime of \mdbp and \mdbb consistently decreases as the number of threads increases, which illustrates the effectiveness of our parallelization technique. Notably, our parallelized algorithms successfully address two challenging inputs where $k=5$ and $\theta=6$ are set on \textit{Google} and \textit{LiveJournal}, which the sequential version of our algorithms failed to solve within 3 hours. This further demonstrates the scalability of our parallelized algorithms.

	\noindent\textbf{Exp9: Memory usage.}
	In this experiment, we evaluate the memory usage of algorithms \mdbb, \mdbp, \mdbase and \mdcadp on five representative graphs. We measure the maximum physical memory (in MB) used by these algorithms with the parameters $k=2$ and $\theta=3$. Similar results can also be observed with other parameters. The results are shown in Fig. \ref{fig:mem}. As observed, all algorithms except \mdcadp exhibit comparable space consumption, which is approximately proportional to the size of input graph. Notably, the space consumed by \mdcadp is higher than other algorithms. This is mainly because \mdcadp requires to fully connect vertices on both sides of a bipartite graph, resulting in increased graph density and a larger search space. These results indicate that our proposed algorithms are space efficient.

	\subsection{Case Study}

	We conduct a case study to evaluate the effectiveness of $k$-MDB in a fraud detection task, following the experimental setup described in \cite{RN63}. We simulates a camouflage attack scenario on the dataset \textit{Appliance} (\url{https://amazon-reviews-2023.github.io/}), which contains 602,777 reviews on 30,459 products by 515,650 users. We inject a fraud block consisting of 500 fake users, 500 fake products, along with 10,000 fake reviews by randomly connect fake users and fake products, and 10,000 camouflage reviews by randomly connect fake users and real products. 
	We adapt \mdbp using the approach outlined in \cite{RN71} to identify the top 2000 $k$-MDBs on the attacked graph, treating all users and products in these $k$-MDBs as fake. For comparison, we implement four baseline methods: maximum biclique \cite{RN13}, maximum $k$-biplex \cite{RN71} and $(\alpha, \beta)$-core\cite{RN83}, where the $(\alpha,\beta)$-core is the largest subgraph in which the vertex degrees in left and right side are no less than $\alpha$ and $\beta$, respectively.  To ensure fairness, we identify the top 2000 largest maximum biclique and maximum $k$-biplex as fake items. We also incorporate a widely used fraud detection method, \fraudar \cite{RN63}, which is a densest-subgraph-based algorithm that identifies suspicious blocks of nodes by iteratively removing edges with the lowest likelihood of being part of a fraudulent structure, thereby resulting in several suspicious subgraphs.
	
	Let $\theta_U$ and $\theta_V$ denote the size thresholds of users and products, respectively. We use $\theta_U$ and $\theta_V$ to constrain the side size of $k$-MDB, maximum $k$-biplex and maximum biclique. Let $n$ denote the number of blocks that \fraudar needs to detect, which is a key parameter that affects the detection accuracy of \fraudar. We test the F1-score ($\frac{2\cdot \mathrm{precision}\cdot \mathrm{recall}}{\mathrm{precision}+\mathrm{recall}}$) of each methods with varying $\theta_V$, $\alpha$ and $n$ from 1 to 7, where we fix $\theta_U=\beta=3$. The results are shown in Fig. \ref{fig:fd}a. As illustrated, when $\theta_V\ge 5$, 2-MDB consistently outperforms other methods. Additionally, the F1-score of 2-MDB rises over 1-MDB while 2-biplex drops below 1-biplex, as $k$-MDB offers finer-grained density relaxation than $k$-biplex with increasing $k$. For $\theta_V\le 4$, biclique and \fraudar outperform $k$-MDB and $k$-biplex, as the latter two are too sparse and skewed to provide meaningful structure. However, the F1-score of biclique declines when $\theta_V\ge 5$, as the strict definition of biclique limits its ability to detect fake items, resulting in lower recall. In contrast, $k$-MDB and $k$-biplex, with their more relaxed definitions, still maintain high F1-scores.

    We also evaluate the running time and the quality of subgraphs generated by each method, utilizing the widely adopted modularity ($\frac{1}{2m}\sum_{u\in U_S,v\in V_S}(1-\frac{d(u)d(v)}{2m})$) metric, which quantifies the strength of internal connectivity relative to a random distribution and serves as a key indicator for assessing the community structure within the subgraph. Fig. \ref{fig:fd}b illustrates the optimal performance of each method across all parameter settings. As observed, 2-MDB achieves the highest modularity value, indicating the high quality of the subgraphs it produces, which consequently results in a higher F1 score for fraud detection. Additionally, although 1-biplex has a F1-score close to 2-MDB, its running time is significantly higher than that of 2-MDB. These results demonstrate the high effectiveness and efficiency of our solutions for fraud detection applications.
	
	\begin{figure}[!t]
		\centering
		
		\begin{minipage}{0.54\linewidth}
			\centering
			\includegraphics[width=1\linewidth,height=2.6cm]{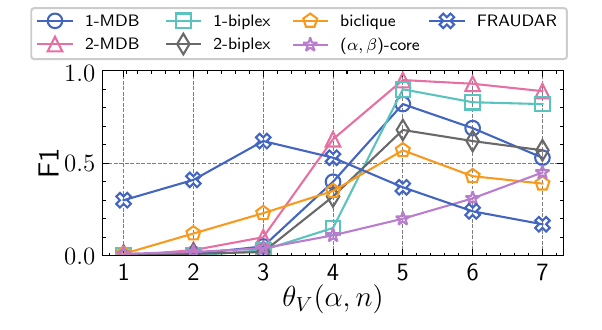}
			
		\end{minipage}
		\hfill
		\begin{minipage}{0.45\linewidth}
			\centering
			\smaller[2]
			\setlength{\tabcolsep}{2.5pt}
			\renewcommand{\arraystretch}{1.2}
			
			\begin{tabular}{lrrr}
				\hline
				Method	&	F1	&	Modularity &	Time(s) \\
				\hline
				1-MDB	&	0.82&	0.33&	0.30\\
				2-MDB	&	\textbf{0.95}&	\textbf{0.45}&	15.44\\
				1-biplex	&	0.92&	0.37&	35.15\\
				2-biplex	&	0.78&	0.24&	1187.82\\
				biclique	&	0.57&	0.15&	0.23\\
				($\alpha,\beta$)-core	&	0.45&	0.04&	0.14\\
				FRAUDAR	&	0.72&	0.27&	17.60\\
				\hline
			\end{tabular}
		\end{minipage}
		\vfill
		\begin{minipage}{0.54\linewidth}
			\small
			\centering \textbf{(a) F1-score with varying parameters}
		\end{minipage}
		\hfill
		\begin{minipage}{0.45\linewidth}
			\small
			\centering \textbf{(b) Optimal performance}
		\end{minipage}
		
		\setlength{\belowcaptionskip}{-0cm}
		\caption{Performance on fraud detection task}
		\label{fig:fd}
		\vspace{-20pt}
	\end{figure}

	\section{Related Work}
	
	\noindent\textbf{Maximum $k$-defective clique search.} As a problem related to our work, maximum $k$-defective clique search has received considerable attention in recent years, resulting in the development of many solutions and algorithms \cite{RN12,RN11,RN10,RN9,RN7,RN6,RN5,RN4,RN3,RN2}. Specifically, Trukhanov et al. \cite{RN10} proposed the first exact algorithm for detecting a maximum $k$-defective clique with a Russian Doll search scheme. Subsequently, Gschwinda et al. \cite{RN11} improved the algorithm with a new preprocessing technique and a better verification procedure. These algorithms are derived from more general clique relaxation search frameworks. There are many studies focusing on the maximum $k$-defective clique. Particularly, Chen et al. \cite{RN6} proposed a branch-and-bound algorithm with a time complexity below $O^*(2^n)$, and introduced a coloring-based upper bound to further improve the efficiency. Gao et al. \cite{RN5} proposed another branch-and-bound algorithm with new vertex and edge reduction rules, achieving better practical performance. Chang \cite{RN3} proposed an algorithm using a non-fully-adjacent-first branching rule, achieving tighter time complexity than the solutions in \cite{RN6}. Subsequently, Chang \cite{RN4} further enhanced the algorithm with an ordering technique, resulting in better time complexity and experimental results. Dai et al. \cite{RN2} proposed a new branch-and-bound algorithm equipped with a pivot-based branching technique and a series of optimization techniques, attaining the best time complexity and experimental results to date. Nevertheless, it is difficult to adapt these algorithms for the $k$-MDB search problem, which is primarily due to the different definitions of maximum.
	
	\noindent\textbf{Maximum biclique search.} The maximum biclique search problem has been extensively studied in recent years \cite{RN22,RN19,RN17,RN16,RN14,RN13,RN23,RN28,RN35}. These studies can be classified into three categories. The first category is maximum edge biclique search, which aims to find a biclique with maximum number of edges. Peeters et al. \cite{RN17} proved the NP-hardness of the problem, and Amb\"{u}hl et al. \cite{RN79} showed the inapproximability of the problem. Shaham et al. \cite{RN19} proposed a probabilistic algorithm for maximum edge biclique search utilizing a Monte-Carlo subspace clustering method. Shahinpour et al. \cite{RN20} provided an integer programming method with a scale reduction technique. Lyu et al. \cite{RN13} proposed a branch-and-bound algorithm with a progressive bounding framework and several graph reduction methods, achieving maximum edge biclique search on billion-scale graphs. The second category is maximum vertex biclique search, which aims to identify a biclique with maximum number of vertices. Garey et al. \cite{RN36} shown that a maximum vertex biclique can be found in polynomial time using a minimum-cut algorithm. The third category is maximum balanced biclique search, which intends to find a maximum biclique with an equal number of vertices on both sides. Yuan et al. \cite{RN28} proposed an evolutionary algorithm for the problem with structure mutation. Wang et al. \cite{RN24} proposed a local search framework with four optimization heuristics. Zhou et al. \cite{RN32} proposed an upper bound propagation served for two algorithms, which are respectively based on branch-and-bound and integer linear programming. Chen et al. \cite{RN23} introduced two efficient branch-and-bound algorithms utilizing a bipartite sparsity measurement, targeting small dense and large sparse graphs, achieving improved time complexities and better practical performance. All these mentioned techniques are tailored for maximum biclique search and cannot be directly used to solve our $k$-MDB search problem. 
	
	
	\vspace{-3pt}
	
	\section{Conclusion}
	
	In this paper, we investigate the problem of finding a maximum $k$-defective biclique in a bipartite graph. We propose two novel algorithms to tackle the problem. The first algorithm employs a newly-designed binary branching strategy that reduces the time complexity to below $O^*(2^n)$, while the second algorithm utilizes a novel pivoting-based branching strategy, achieving a notably lower time complexity. Additionally, we propose a series of nontrivial optimization techniques to further improve the efficiency of our algorithms. Extensive experiments conducted on 10 real-word graphs demonstrate the notable effectiveness and efficiency of our algorithms and optimization techniques.

	\bibliographystyle{ACM-Reference-Format}
	\balance\bibliography{refs}
	
	\let\baselinestretch\savedbaselinestretch
	
\end{document}